\def\polylog{\mathop{\rm polylog}}
\def\wgt{\mathop{\rm wgt}}
\def\ord{\mathop{\rm ord}}
\def\rank{\mathop{\rm rank}}
\def\diag{\mathop{\rm diag}}
\def\css{\mathop{\rm CSS}}
\def\tr{\mathop{\rm tr}\nolimits}

\def\LL{\mathop{\mathrm L}\nolimits}
\def\RR{\mathop{\mathrm R}\nolimits}
\def\lp{\mathop{\rm LP}}
\def\hp{\mathop{\rm HP}}
\def\gb{\mathop{\rm GB}}
\def\CC{\mathop{\cal C}}

\documentclass[aps,pra,tightenlines,reprint,floatfix,%
notitlepage,longbibliography]{revtex4-2}

\usepackage[colorlinks,citecolor=cyan,driverfallback=pdftex]{hyperref}
\usepackage{url}
\usepackage{amsmath}
\usepackage{amssymb}
\usepackage{xcolor}

\usepackage{soul} 
\setulcolor{blue}
\setstcolor{red}
\definecolor{lightblue}{rgb}{.90,.95,1}
\sethlcolor{lightblue}

\usepackage{amsfonts}
\usepackage{graphicx}
\graphicspath{{./}{./figs/}}

\usepackage[inline]{enumitem}
\setlist{nosep}

\usepackage{amsthm}
\newtheorem{theorem}{Theorem}
\newtheorem{statement}[theorem]{Statement}

\newtheorem{lemma}[theorem]{Lemma}
\newtheorem{example}[theorem]{Example}

\def\bs#1{\boldsymbol{#1}}

\addtocounter{dbltopnumber}{2}
\addtocounter{topnumber}{2}

\begin{document}
\title{Quantum two-block group algebra codes}

\author{Hsiang-Ku Lin}
\affiliation{Department of Physics \& Astronomy, University of
  California, Riveside, California 92521 USA}

\author{Leonid P. Pryadko}
\affiliation{Department of Physics \& Astronomy, University of
  California, Riveside, California 92521 USA}
\date{\today}

\begin{abstract}
  We consider quantum two-block group algebra (2BGA) codes, a
  previously unstudied family of smallest lifted-product (LP) codes.
  These codes are related to generalized-bicycle (GB) codes, except a
  cyclic group is replaced with an arbitrary finite group, generally
  non-abelian.  As special cases, 2BGA codes include a subset of
  square-matrix LP codes over abelian groups, including quasi-cyclic
  codes, and all square-matrix hypergraph-product codes constructed
  from a pair of classical group codes.  We establish criteria for
  permutation equivalence of 2BGA codes and give bounds for their
  parameters, both explicit and in relation to other quantum and
  classical codes.  We also enumerate the optimal parameters of all
  inequivalent connected 2BGA codes with stabilizer generator weights
  $W\le 8$, of length $n\le 100$ for abelian groups, and $n\le 200$
  for non-abelian groups.
\end{abstract}
\maketitle
\section{Introduction}
\label{sec:introduction}

Recent years have seen a substantial progress in theory of quantum
low-density parity-check (LDPC) codes\cite{Evra-Kaufman-Zemor-2020,%
  Hastings-Haah-ODonnell-2020,Panteleev-Kalachev-2020,%
  Breuckmann-Eberhardt-2020,Panteleev-Kalachev-2021}.  Generally, any
code family with bounded-weight stabilizer generators and distance
scaling logarithmically or faster with the block length has a finite
fault-tolerant threshold to scalable error
correction\cite{Dennis-Kitaev-Landahl-Preskill-2002,%
  Kovalev-Pryadko-FT-2013,Dumer-Kovalev-Pryadko-bnd-2015}.  Unlike in
the case of classical LDPC codes\cite{Gallager-book-1963,%
  Chung-Forney-Richardson-Urbanke-2001} where random matrices are
commonly used to define the code, due to a commutativity constraint,
an algebraic ansatz is required in the case of quantum LDPC codes.
For over a decade, no construction was known to give distances larger
than a square root of the block size $n$, up to a polylogarithmic
factor\cite{kitaev-anyons,Dennis-Kitaev-Landahl-Preskill-2002,%
  Freedman-Meyer-Luo-2002, Tillich-Zemor-2009,
  Kovalev-Pryadko-Hyperbicycle-2013,
  Guth-Lubotzky-2014,Evra-Kaufman-Zemor-2020,Zeng-Pryadko-2018,%
  Zeng-Pryadko-hprod-2020,*Zeng-Pryadko-erratum-2022,Kaufman-Tessler-2021}.
The barrier was broken by Hastings, Haah, and
O'Donnell\cite{Hastings-Haah-ODonnell-2020} who demonstrated a quantum LDPC
code family with the distance scaling as
$\mathcal{O}(n^{3/5}/\polylog n)$.  Soon followed related
constructions\cite{Panteleev-Kalachev-2020,%
  Breuckmann-Eberhardt-2020}, with Panteleev and
Kalachev\cite{Panteleev-Kalachev-2021} finally proving the existence
of asymptotically good bounded-stabilizer-generator-weight quantum
LDPC codes, with non-zero asymptotic relative distances for any
asymptotic rate $R<1$.  Unfortunately, the constructions in
Refs.~\onlinecite{Hastings-Haah-ODonnell-2020,%
  Panteleev-Kalachev-2020,Breuckmann-Eberhardt-2020,Panteleev-Kalachev-2021}
tend to give rather long codes, and the lower
bound\cite{Panteleev-Kalachev-2021} for the row weight to give
asymptotically good quantum codes is also very large.

In Ref.~\cite{Wang-Pryadko-2022}, in an attempt to construct shorter
quantum LDPC codes with large distances, one of us studied a class of
generalized bicycle (GB)
codes\cite{Kovalev-Pryadko-Hyperbicycle-2013,Panteleev-Kalachev-2019}.
These are index-two quantum quasi-cyclic (qQC) codes, a special case
of qQC
codes\cite{Hagiwara-Imai-2007,Kasai-Hagiwara-Imai-Sakaniwa-2012,%
  Kasai-Hagiwara-Imai-Sakaniwa-2012,%
  Galindo-Hernando-Matsumoto-2018,Panteleev-Kalachev-2019} where the
general upper distance bound related to the number of blocks does not
apply.  Important advantages of GB codes are overcomplete set of
minimum-weight stabilizer generators which may improve their
performance in the fault-tolerant setting, and their regular structure
which simplifies implementation and iterative
decoding\cite{Panteleev-Kalachev-2019,Raveendran-Vasic-2021}.
Furthermore, GB codes include\cite{Wang-Pryadko-2022} codes with
linear distance scaling, unlike, e.g., the hypergraph-product
codes\cite{Tillich-Zemor-2009}, where the distance can never exceed a
square root of the block length.  However, the regular structure of
the corresponding matrices also implies\cite{Wang-Pryadko-2022} that a
GB code with row weight $W$ can be mapped to a code local in dimension
$D\le W-1$, which implies a power-law upper bound on the distance,
$d\le {\cal O}(n^{1-1/D})$, see
Ref.~\onlinecite{Bravyi-Terhal-2009,*Bravyi-Poulin-Terhal-2010}.
Numerically, it appears that fixed-weight GB codes have distance
scaling as $A(W)n^{1/2}$, where $A(W)$ is an increasing function of
the weight $W$, although a power-law scaling $d={\cal O}(n^{\alpha})$
with $\alpha>1/2$ but close to the lower bound cannot be
excluded\cite{Wang-Pryadko-2022}.

The goal of this work is to explore a class of codes similar to GB
codes, where  more general symmetry groups are used instead of cyclic
groups\cite{[{Some of the present results have been announced %
    previously, see }]Wang-Lin-Pryadko-2023}.
These codes are a special case of two-block CSS
codes\cite{Kovalev-Pryadko-Hyperbicycle-2013}, and also are the
smallest lifted-product (LP) codes\cite{Panteleev-Kalachev-2021}.  In
fact, this work was inspired by the LP codes construction, along with
the related work by the same authors on two-block codes based on
abelian group algebras\cite{Kalachev-Panteleev-2020}.  Our main reason
to study these \emph{two-block group algebra} (2BGA) codes, especially
in the non-abelian case, is that general upper distance
bound\cite{Panteleev-Kalachev-2021} for LP codes does not apply in the
two-block case, and neither do the upper distance
bounds\cite{Wang-Pryadko-2022} for GB codes with row weight $W$ since
more general abelian or non-abelian groups do not give matrices with
structure as regular as that of the circulant matrices.  On the other
hand, most of the advantages of the GB codes remain.  In particular,
these more general codes also have naturally overcomplete sets of
minimum-weight stabilizer generators, which is expected to improve
their performance in the fault-tolerant setting.

The outline of the rest of the paper is as follows.  We give some
background information in Sec.~\ref{sec:notations}.  In
Sec.~\ref{sec:two-block}, we discuss general properties of quantum CSS
codes constructed from two square commuting matrices.  In
Sec.~\ref{sec:construction} we give the construction of 2BGA codes and
analyze their properties, and in Sec.~\ref{sec:numerical} discuss the
parameters of 2BGA codes constructed numerically.  Finally, we give
the conclusions in Sec.~\ref{sec:conclusions}.  More technical proofs
for Sections~\ref{sec:two-block} and \ref{sec:construction} are
collected in Appendices \ref{sec:proofs-III} and \ref{sec:proofs-IV},
respectively.  Appendix \ref{sec:examples} gives additional examples
of index-4 qQC 2BGA codes constructed from groups
$C_{mh}$ and $D_{m}$.

\section{Notations and known facts}
\label{sec:notations}

\subsection{Classical codes}

A classical $q$-ary error-correcting code with parameters $(n,K,d)_q$
is a set of $K$ strings of length $n$ in an alphabet with $q$ distinct
characters, where any two strings differ in $d$ or more
positions\cite{MS-book}.  In a \emph{linear code} $C\subset F^n$ using
as the alphabet a finite Galois field $F\equiv \mathbb{F}_q$, where
$q=p^m$ is a power of a prime $p$, the field characteristic, the
strings in the code form a linear space of dimension $k$, so that
$K=q^k$.  The parameters of such a code are denoted $[n,k,d]_q$, where
the distance $d$ is the minimum Hamming weight of a non-zero vector in
the code.  For a trivial code with $k=0$ (empty set of non-zero
vectors), we set $d$ equal to infinity.

Rows of a \emph{generator matrix} $G$ of a linear code $C\equiv C_G$
are non-zero vectors in the code which include a complete basis, so
that any vector of the code can be written as a linear combination of
rows of $G$; evidently, $\rank G=k$.  The code $C^\perp$ \emph{dual}
to $C$ is formed by all vectors in $F^n$ orthogonal to the vectors in
$C$.  A generator matrix $H$ of the code $C_G^\perp$ is called a
parity check matrix of the original code $C_G$, it satisfies the duality relation
\begin{equation}
  \label{eq:ranks}
  GH^T=0,\quad \rank G+\rank H=n.
\end{equation}

Given a string $\bs c\in{F}^n$, denote
${\cal V}\equiv [n]\equiv \{1,2,\ldots,n\}$ the set indexing the individual
characters.  For any \emph{index set} ${\cal I}\subseteq {\cal V}$ of
length $|{\cal I}|=r$, let $\bs c[{\cal I}]\in{F}^r$ be a substring of
$\bs c$ with the characters in all positions $i\not\in {\cal I}$
dropped.  We say that $\bs c[\mathcal{I}]$ is the string $\bs c$
\emph{punctured} outside $\mathcal{I}$.  Similarly, for an $n$-column
matrix $G$, the punctured matrix $G[{\cal I}]$ is formed by the rows
of $G$ punctured outside $\mathcal{I}$. If $C=C_G$ is an $F$-linear
code with the generating matrix $G$, then the code of length
$|{\cal I}|$ with the generating matrix $G[{\cal I}]$ is the code
punctured outside ${\cal I}$,
${C}_{\rm p}({\cal I})\equiv\{c[{\cal I}]\,|\, c\in{C}\}$.

The \emph{shortened} code ${C}_{\rm s}({\cal I})$ is formed similarly,
except only from the codewords supported inside ${\cal I}$,
${C}_{\rm s}({\cal I})=\{\bs c[{\cal I}]\,|\,\bs
c=(c_1,c_2,\ldots,c_n)\in {C}$ and $c_i=0$ for each
$i\not\in {\cal I}\}$.  The dual of a punctured code
${C}_{\rm p}({\cal I})$ is the shortened dual code,
$[{C}_{\rm p}({\cal I})]^\perp=({C}^\perp)_{\rm s}({\cal I})$.  To
express this relation in terms of matrices, consider a pair of
mutually dual matrices in Eq.~(\ref{eq:ranks}) and a code
${C}\equiv {C}_G={C}_H^\perp$.  Denote a generator matrix of the
shortened code ${C}_{\rm s}({\cal I})$ as $G_{\cal I}$.  Duality
between the punctured original and the shortened dual codes implies
that the corresponding generator matrices $G_{\cal I}$ and
$H[{\cal I}]$ are also mutually dual\cite{MS-book},
\begin{equation}
  \label{eq:dual-shortening}
  H[{\cal I}]\, G_{\cal I}^T=0,\quad \rank G_{\cal I}+\rank H[{\cal I}]=|{\cal I}|.
\end{equation}
Similarly, $H_{\cal I}$ is a dual of the punctured matrix $G[{\cal I}]$.

Relevant for this work are left and right group codes constructed in a
group algebra\cite{Milies-2019,Ferraz-Milies-Taufer-2021}.  Namely,
for a given finite field $F$ and a finite group $G$ of order
$|G|=\ell$, we consider the group algebra (a ring) $F[G]$ defined as an
$F$-linear space of all formal sums
\begin{equation}
  \label{eq:algebra-element}
  x\equiv \sum_{g\in G}x_g g,\quad x_g\in F,
\end{equation}
where group elements $g\in G$ serve as basis vectors,
equipped with the product naturally associated with the group
operation,
\begin{equation}
  \label{eq:FG-product}
  ab=\sum_{g\in G}\biggl(\sum_{h\in G} a_h b_{h^{-1}g}\biggr) g, \quad a,b\in F[G].
\end{equation}
Evidently, Eq.~(\ref{eq:algebra-element}) defines a one-to-one map
between any vector $\bs x\in F^\ell$ with coefficients $x_g$ labeled
by group elements and a group algebra element $x\in F[G]$, and a similar
map between sets of vectors and sets of group algebra elements.  A
left $G$-code in $F^\ell$ is a map of a left ideal $J_L$ in the ring
$F[G]$, defined as an $F$-linear space of elements of $F[G]$ such that for
any $x\in J_L$ and any $r\in F[G]$, $rx\in J_L$.  A right $G$-code is
defined similarly in terms of a right ideal $J_R$, with the opposite
order in the product, $xr\in J_R$ for any $x\in J_R$ and any
$r\in F[G]$.

The structure of ideals in $F[G]$ is particularly simple if
characteristic of the field does not divide the group size,
$\gcd(p,\ell)=1$.  Then, according to Maschke's theorem, the group algebra is
semisimple, and any ideal is a principal ideal generated by an
idempotent, e.g., $J_L=F[G]\cdot f_J$ for a left ideal,
with idempotent $f_J^2=f_J\in J_L$, and similarly,
$J_R=e_J\cdot F[G]$ for a right ideal, with idempotent
$e_J^2=e_J\in J_R$ (see, e.g.,
Corollary 2.2.5 in Ref.~\onlinecite{Drozd-Kirichenko-book-1994}).

The usual inner product in $F^\ell$ is related to the group trace
by a linear
map\cite{Borello-delaCruz-Willems-2022} $\,\bs{\widehat{\ }}:F[G]\to
F[G]$,
\begin{equation}
\widehat a\equiv\sum_{g\in G}a_{g^{-1}} g=\sum_{g\in G}a_g g^{-1}.\label{eq:hat-operator}
\end{equation}
Namely, for any $\bs a, \bs b\in F^\ell$, and the corresponding group
algebra elements $a, b\in F[G]$,
\begin{equation}
  \label{eq:inner}
  {\bs a}\cdot {\bs b}\equiv \sum_{g\in G}a_g b_g
  =\tr_G(\widehat{a} b)=\tr_G(b\widehat a).
\end{equation}
As a reminder, the group trace is defined as the coefficient of the
group identity element $1\in G$: for any $a\in F[G]$,
$\tr_G(a)\equiv a_{1}\in F$.

Given a right group code in a semisimple group algebra $F[G]$, that
is, a right ideal $J_R\equiv \widehat a\cdot F[G]$
generated by an element $\widehat a\in F[G]$, any group algebra
element $x$ corresponding to a vector $\bs x$ in the orthogonal code
satisfies\cite{Borello-delaCruz-Willems-2022} the equation $x a =0$.
If we denote an idempotent $e_a^2=e_a\in F[G]$ such that $e_aa=a$, the
solution of the orthogonality equation is the left ideal
$J_L\equiv F[G]\cdot (1-e_a)$.

\subsection{Quantum CSS codes}

A quantum Calderbank-Shor-Steane (CSS)
code\cite{Calderbank-Shor-1996,*Steane-1996} ${Q}=\css(H_X,H_Z)$ with parameters
$[[n,k,d]]_q$ over a Galois field $F$ is isomorphic to a direct sum of
an $X$- and a $Z$-like codes,
\begin{equation}
  \label{eq:css-code}
Q={Q}_X\oplus {Q}_Z={C}_{H_Z}^\perp/{C}_{H_X}\oplus
  {C}_{H_X}^\perp/{C}_{H_Z},
\end{equation}
where each term in the right-hand side (r.h.s.) is a quotient of two
linear spaces in $F^n$, and rows of the stabilizer generator matrices
$H_X$ and $H_Z$ must be orthogonal,
\begin{equation}\label{eq:CSS-orthogonality}
  H_XH_Z^T=0.
\end{equation}
Explicitly, e.g., elements of ${Q}_X$ are equivalence classes of
vectors orthogonal to the rows of the matrix $H_Z$, with any two
vectors whose difference is a linear combination of the rows of $H_X$
identified.  Vectors in the same class are called mutually degenerate,
while vectors in the class of the zero vector are called trivial.  The
codes ${Q}_X$ and ${Q}_Z$ have $q^k$ degeneracy classes
each, where
\begin{equation}
  \label{eq:k-CSS}
  k=n-\rank H_X-\rank H_Z
\end{equation}
is the quantum code dimension.  The distance of the code is
$d\equiv \min(d_X,d_Z)$, where the two CSS distances,
\begin{equation}
  \label{eq:d-CSS}
  d_X=\min_{\bs c\in C_{H_Z}^\perp\setminus C_{H_X}}\wgt \bs c,\quad
  d_Z=\min_{\bs c\in C_{H_X}^\perp\setminus C_{H_Z}}\wgt \bs c,
\end{equation}
are the minimum weights of non-trivial vectors (any representative) in
${C}_{H_Z}^\perp$ and ${C}_{H_X}^\perp$, respectively.  A set of
logical operators' representatives in $Q_X$ and $Q_Z$ can be chosen to
form $k$ canonically conjugate pairs.  Equivalently, logical generator
matrices $L_X$ and $L_Z$ with $k$ rows each can be constructed such
that
\begin{equation}
  \label{eq:logical}
  L_XH_Z^T=0,\quad L_ZH_X^T=0,\quad L_XL_Z^T=I_k,
\end{equation}
where $I_k$ is a $k\times k$ identity matrix.

Physically, a quantum code operates in a Hilbert space
${\cal H}_q^{\otimes n}$ associated with $n$ quantum-mechanical
systems, Galois-qudits\cite{eczoo-galois-qudits}, with $q$ states
each, and a well defined basis of $X$ and $Z$ operators acting in
${\cal H}_q^{\otimes
  n}$\cite{Ketkar-Klappenecker-Kumar-Sarvepalli-2006}.  Elements of
the codes $C_{H_X}$ and $C_{H_Z}$ correspond to $X$- and $Z$-
operators in the stabilizer group $\mathcal{S}$ acting in the Hilbert
space.  Generators of ${\cal S}$ must be measured frequently during the
operation of the code; generating matrices $H_X$ and $H_Z$ with
smaller row weights result in codes which are easier to implement in
practice.  Orthogonality condition (\ref{eq:CSS-orthogonality})
ensures that the stabilizer group is abelian.  Non-trivial vectors in
${Q}_X$ and ${Q}_Z$ correspond to $X$ and $Z$ logical operators,
respectively.  Codes with larger distances have logical operators
which involve more qudits; such codes typically give better
protection.

More generally, a CSS \emph{subsystem}
code\cite{Poulin-subs-2005,Bacon-subs-2006}
\begin{equation}
\css(G_X,G_Z)=Q_X\oplus Q_Z\label{eq:subsystem-CSS}
\end{equation}
can be defined by two $n$-column gauge generator matrices $G_X$
and $G_Z$ whose rows are not necessarily orthogonal.  Such a code can
be constructed from a regular CSS code (\ref{eq:css-code}) of
dimension $k_{\rm orig}=k+p_\star$ by selecting $p_\star\le k_{\rm orig}$ logical
operator pairs and adding the corresponding rows [forming matrices
$L_X'$, $L_Z'$ such that $L_X' (L_Z')^T=I_{p_\star}$] to the rows of the CSS
generator matrices,
\begin{equation}
G_X=U_X{H_X\choose L_X'},\quad G_Z=U_Z{H_Z\choose L_Z'},\label{eq:subsystem-CSS-constr}
\end{equation}
where $U_X$ and $U_Z$ are invertible matrices corresponding to
arbitrary row transformations, and the subscript in $p_\star$ is to
disambiguate with the field characteristic $p$.  Respectively, for a
CSS subsystem code,
\begin{equation}
k=n-\rank G_X-\rank G_Z+\rank (G_X G_Z^T),\label{eq:subsystem-k}
\end{equation}
and its distance, e.g., for the subcode $Q_Z$,
\begin{equation}
  \label{eq:subsystem-d}
  d_Z=\min_{\bs c\in C_{H_X}^\perp\setminus C_{G_Z}}\wgt \bs c=
  \min_{\bs c\in C_{G_X}^\perp\setminus C_{G_Z}}\wgt \bs c.
\end{equation}
Prominent example of subsystem codes are erasure codes obtained when
matching sets of columns are removed from the stabilizer generator
matrices $H_X$ and $H_Z$.  Equivalently, with ${\cal I}$ the index set
of the remaining columns, a subsystem erasure code has the punctured
stabilizer group $\mathcal{S}_{\rm p}({\cal I})$.

\section{Two-block CSS codes}
\label{sec:two-block}

Here we discuss general properties of two-block CSS
codes\cite{Kovalev-Pryadko-Hyperbicycle-2013} with generator matrices
in the form
\begin{equation}
  \label{eq:css-blocks}
  H_X=(A,B),\quad H_Z^T={B\choose -A},
\end{equation}
where $A,B\in M_{\ell}(F)$ are square commuting matrices of size
$\ell\times \ell$ with elements in a Galois field $F$.  The
commutativity is important, since it guarantees the CSS orthogonality
condition (\ref{eq:CSS-orthogonality}).

An important tool in analyzing the parameters of such codes will be
the subsystem \emph{block-erasure} code $\css(A,B^T)$ and its
CSS-dual, obtained by erasing the qudits in the right
and left blocks, respectively.  We will denote the common parameters
of these codes as
\begin{equation}
  \label{eq:subsystem-AB}
  [[\ell,k_{\rm S},d_{\rm S}]]_q,\quad\text{and}\quad p_\star\equiv\rank (AB),
\end{equation}
with $p_\star$ the number of gauge qudits,
cf.~Eqs.~(\ref{eq:subsystem-CSS-constr}) and (\ref{eq:subsystem-k}).%

\subsection{Code dimension}
\label{sec:tb-dimension}
Given a square matrix $A\in M_\ell(F)$ of size
$\ell$ with elements in the Galois field $F$, consider
size-$\ell$ idempotent matrices $E_A$ and $F_A$
of the same rank as $A$, such that
\begin{equation}
  \label{eq:idempotent-EA-FA}
  E_A^2=E_A,\quad F_A^2=F_A,\quad E_AA=AF_A=A.
\end{equation}
While such matrices are not unique, they can always be constructed
from the decomposition $A=U_AD_AV_A$, where $U_A,V_A\in M_\ell(F)$
are square invertible matrices, and
$D_A=\diag(1,\ldots,1,0,\ldots,0)\in M_\ell(F)$ has exactly $\rank A$
non-zero elements along the diagonal.  Namely, we may choose
\begin{equation}
  E_A\equiv U_AD_AU_A^{-1},\quad  F_A\equiv V_A^{-1}D_AV_A.\label{eq:EA-FA-matrices}
\end{equation}
With idempotent matrices (\ref{eq:idempotent-EA-FA}), it is now easy
to calculate the ranks of matrices (\ref{eq:css-blocks}).  Indeed, row
and column transformations give (this is a simplified version of more
general expressions in Ref.~\onlinecite{Meyer-1970,Meyer-1973})
\begin{eqnarray}
  \nonumber
  \rank H_X&=&\rank \left(
               \begin{array}[c]{cc}
                  A&E_AB\\ 0&(I-E_A)B
               \end{array} \right)\\    \nonumber
           &=&\rank(A)+\rank (I-E_A)B,  \\
            &=&\rank A+\rank B-\rank (E_AB),
               \label{eq:rankHx}
\end{eqnarray}
where we also expressed $\rank B$ with the help of a similar
decomposition, $\rank B=\rank (E_AB)+\rank(I-E_A)B$.  Similarly, for
the other matrix we get
\begin{equation}
  \rank H_Z= 
  \rank A+\rank B-\rank (BF_A).
  \label{eq:rankHz}
\end{equation}
We have, e.g., $\rank E_AB\ge \rank E_ABA=\rank AB=p_\star$.  For a given
set of idempotent matrices (\ref{eq:idempotent-EA-FA}), introduce
non-negative \emph{rank defect} parameters $\delta_X\ge 0$ and $\delta_Z\ge0$,
\begin{equation}
  \label{eq:rank-defect}
  \rank E_AB
  \equiv p_\star+\delta_X, \quad
  \rank BF_A
  \equiv   p_\star+\delta_Z,
\end{equation}
where $ p_\star\equiv \rank AB$ is the number of gauge qubits in the
subsystem code $\css(A,B^T)$, see Eq.~(\ref{eq:subsystem-AB}).  While
rank defects are introduced with respect to a specific set of
idempotents $E_A$ and $F_A$, Eqs.~(\ref{eq:rankHx}) and
(\ref{eq:rankHz}) guarantee that they are, in fact, independent of the
choice of idempotents in Eq.~(\ref{eq:idempotent-EA-FA}).  Moreover,
the same parameters can be also introduced in terms of similarly
defined idempotent matrices associated with the matrix $B$,
\begin{equation}
  \label{eq:rank-defect-B}
  \rank E_BA= p_\star+\delta_X, \quad
  \rank AF_B=  p_\star+\delta_Z.
\end{equation}
Physically, $\delta_X$ and $\delta_Z$ are the numbers of rows in $H_X$
and $H_Z$, respectively, which give non-trivial linearly-independent
contributions to the centers of both ${\cal S}[{\cal I}_L]$ and
${\cal S}[{\cal I}_R]$, the stabilizer group punctured to
individual blocks.  Combining the obtained expressions
(\ref{eq:rankHx}) and (\ref{eq:rankHz}) with
Eq.~(\ref{eq:subsystem-k}) and the definitions (\ref{eq:rank-defect}),
we have
\begin{equation}
  \label{eq:rank}
  \rank H_X=\ell-k_{\rm S}-\delta_X,\quad
  \rank H_Z=\ell-k_{\rm S}-\delta_Z,
\end{equation}
which gives for the original two-block code (\ref{eq:css-blocks}),
\begin{equation}
k=2k_{\rm S}+\delta_X+\delta_Z. \label{eq:two-block-k}
\end{equation}
Most generally, $\delta_X\neq \delta_Z$, and these parameters are
non-negative.  However, a rank defect is guaranteed to vanish with an
additional commutativity condition (see Appendix \ref{sec:proofs-III}
for the proofs):
\begin{statement}
  \label{th:rank-defect-zero}
  If $E_A$ commutes with $B$, then $\delta_X=0$.  Similarly, if $F_A$
  commutes with $B$, then $\delta_Z=0$.
\end{statement}
A similar statement is also valid in terms of the idempotents $E_B$
and $F_B$, e.g., $\delta_X=0$ if $E_BA=AE_B$.

Another special case is when there exists an invertible matrix $S$
which can simultaneously transform both matrices $A$ and $B$ into
their transposed,
\begin{equation}
SAS^{-1}=A^T,\quad SBS^{-1}=B^T.\label{eq:sim-transpose}
\end{equation}
Here, with any choice of $E_A$, we can take $F_A^T=S E_AS^{-1}$,
which, with Eq.~(\ref{eq:rank-defect}), immediately gives
$\delta_X=\delta_Z$, not necessarily zero.  This gives
\begin{statement}
  \label{th:rank-defect-equal}
  If both matrices $A$ and $B$ can be simultaneously transformed into
  their respective transposed, see Eq.~(\ref{eq:sim-transpose}), then
  $\delta_X=\delta_Z$.
\end{statement}
While the condition may appear unnatural, as we discuss below, it is
satisfied for abelian 2BGA codes\cite{Kalachev-Panteleev-2020}.

To summarize this section, most generally, $\delta_X\neq \delta_Z$,
and $\rank H_X\neq \rank H_Z$, so that the dimension of a two-block
code does not have a particular parity.  However, under conditions of
Statement \ref{th:rank-defect-zero} or Statement
\ref{th:rank-defect-equal}, we get $\rank H_X=\rank H_Z$, and code
dimension $k$ even.  Furthermore, under conditions of Statement
\ref{th:rank-defect-zero} we have $\delta_X=\delta_Z=0$, so that
$k=2k_{\rm S}$, exactly twice the dimension of the block-erasure
subsystem code $\css(A,B^T)$.

\subsection{Upper distance bounds}
\label{sec:tb-upper-d}
The same idempotent matrices can be used to analyze the structure of
the codewords.  Most generally, one can expect a given non-trivial
codeword $\bs c_Z\equiv {\bs u\choose \bs v}$ either to be equivalent
to such a codeword with only one of the components non-zero, or not,
in which case any equivalent codeword has both $\bs u$ and $\bs v$
non-zero.  Unlike the cases of HP or GB
codes\cite{Tillich-Zemor-2009,Wang-Pryadko-2022}, for two-block codes
with $\delta_X>0$ or $\delta_Z>0$, it is not possible to choose a full
set of mutually non-degenerate and independent codewords in the former
class.  Nevertheless, the corresponding projections can be used to
construct upper bounds on the distances.  Specifically, consider two
reduced-dimension codes $Q_\mu'\equiv \css(H_X^{(\mu)},H_Z)$,
$\mu\in \{L,R\}$, with
\begin{equation}
  \label{eq:mat-HXi}
  H_X^{(L)}=\left(
    \begin{array}[c]{cc}
      A&B\\ 0&I-E_A
    \end{array}
  \right),\;\;
  H_X^{(R)}=\left(
    \begin{array}[c]{cc}
      A&B\\ I-E_B&0
    \end{array}
  \right),
\end{equation}
where additional rows guarantee that $Z$-logical operators can be
chosen to be supported on one block only, two single-block
$Z$-shortened codes, $Q_L''\equiv\css\biglb(A,(H_Z)_L\bigrb)$ and
$Q_R''\equiv\css\biglb(B,(H_Z)_R\bigrb)$, with
\begin{equation}
  \label{eq:mat-HZi-shortened}
  (H_Z)_L^T=B(I-F_A),\quad (H_Z)_R^T=A(I-F_B),
\end{equation}
and two classical codes $C_L$, $C_R$ with parity check matrices, respectively,
\begin{equation}
  \label{eq:mat-HL-HR}
  H_L\equiv   {A\choose E_B},\quad
  H_R\equiv   {B\choose E_A}.
\end{equation}
As detailed in Appendix \ref{sec:proofs-III}, for a chosen
$\mu\in \{L,R\}$, these definitions correspond to a series of
subsequent restrictions on $Z$ codewords, and we get:
\begin{statement}
  \label{th:d-upper-chain}
  For a given two-block code $Q$ and a chosen $\mu\in \{L,R\}$,
  consider quantum codes $Q_\mu'$ and $Q_\mu''$, and a classical code
  $C_\mu$.  Distances of these codes satisfy
  \begin{equation}
    \label{eq:d-upper-chain}
    d_Z\equiv d_Z(Q)\stackrel{({\rm a})}{\le} d_Z(Q_\mu')
    \stackrel{({\rm b})}{\le } d_Z(Q_\mu'') \stackrel{({\rm c})}{\le} d(C_\mu).
  \end{equation}
\end{statement}
This implies the inequality $d_Z\le d_Z(Q_\mu'')$, a special case of
$Z$-shortening lemma from Ref.~\onlinecite{Zeng-Pryadko-hprod-2020}.

A particularly simple upper bound for the distance $d_Z(Q_L'')$, and
thus for the distance of the original two-block code
(\ref{eq:css-blocks}), is obtained when matrix $A$ is block-diagonal
(see the proof in Appendix \ref{sec:proofs-III}):
\begin{statement}
  \label{th:upper-d-block-diagonal}
  Suppose matrix $A$ is block-diagonal with the maximum block size
  $m$, and the code $Q_L''$ is non-trivial, $k(Q_L'')>0$.  Then the
  distance $d_Z(Q_L'')\le m$.
\end{statement}
Evidently, when matrix $B$ is block-diagonal, a similar bound also
exists for $d_Z(Q_R'')$.

\subsection{Lower distance bounds}
\label{sec:tb-lower-d}
Best known are the usual CSS
bounds,
\begin{equation}
  \label{eq:lower-d-CSS-bound}
  d_Z\ge d(C_{H_X}^\perp), \quad d_X\ge d(C_{H_Z}^\perp).
\end{equation}
However, since the rows of $H_X$ and $H_Z$ are mutually orthogonal, we
have, e.g., $d(C_{H_X}^\perp)\le d(C_{H_Z})\le W_Z$, the minimum row
weight of the matrix $H_Z$.  Since our main interest is in
highly-degenerate quantum LDPC codes with bounded stabilizer weights
and diverging distances, the CSS bounds (\ref{eq:lower-d-CSS-bound})
are not very useful.

Here we construct lower bounds for the distance in terms of the
distances of single-block codes.  It is easy to see that the
$Z$-punctured stabilizer codes
\begin{equation}
  \css\biglb((1-E_B)A,B^T\bigrb) \text{ and }
  \css\biglb((1-E_A)B,A^T\bigrb)\label{eq:Z-punctured}
\end{equation}
both have the dimension $k_{\rm S}+\delta_X$.  In the special case
$\delta_X=0$, this is the same as for single-block erasure subsystem
codes (\ref{eq:subsystem-AB}), and the $Z$-distances are also the same
as $d_Z(A,B^T)$ and $d_Z(B,A^T)$, respectively.  This condition also
guarantees that any non-trivial $Z$-codeword in one of these codes
becomes a non-trivial codeword in the original two-block code after it
is padded with zeros.  The additional condition $\delta_Z=0$
guarantees that the full set of linearly-independent $Z$-codewords of
the two-block code can be constructed this way, which coincides with
the condition of $Z$-puncturing lemma from
Ref.~\onlinecite{Zeng-Pryadko-hprod-2020}.  With the help of the fact
that the two single-block erasure codes with parameters
(\ref{eq:subsystem-AB}) are related by CSS conjugation, we obtain a
simple lower distance bound
\begin{statement}
  \label{th:d-lower-puncturing}
  Suppose both rank defects in Eq.~(\ref{eq:rank-defect}) are zero,
  $\delta_X=\delta_Z=0$.  Then,
\begin{equation}
  \label{eq:d-lower-puncturing}
 d\ge d_{\rm S},\quad d_{\rm S}\equiv d(A,B^T).
\end{equation}
\end{statement}
We notice that under the conditions of Statement
\ref{th:d-lower-puncturing}, e.g., the $Z$-shortened code $Q_L''$ in
Statement \ref{th:d-upper-chain} can also be seen as a gauge-fixed
subsystem code $\css(A,B^T)$.  However, this particular gauge fixing
may result in the increased $d_Z$.  Therefore, we should not
necessarily expect the lower bound (\ref{eq:d-lower-puncturing}) to
saturate, except when $p=0$, or, equivalently, $AB=0$, in which case
the erasure code $\css(A,B^T)$ is also a stabilizer code.

We finish this section with two more general expressions relating the
distance $d_Z$ of a two-block code with those of auxiliary quantum
codes of smaller dimension.  Namely, any non-trivial vector
$\bs c_Z={\bs u\choose \bs v}$ is either degenerate to a solution with
$\bs u$ non-zero and $\bs v=0$, or to a solution with any $\bs u$ and
$\bs v\ncong 0$, but not both.  This and an equivalent construction
with $\bs u$ and $\bs v$ interchanged give two generalizations of
Statement 1 from Ref.~\onlinecite{Wang-Pryadko-2022},
\begin{equation}
  \label{eq:dZ-quantum-identity-A}
  d_Z=\min \left\{ d_Z(H_X^{(\mu)},H_Z),d_Z(H_X,H_Z^{(\mu)})\right\},
\end{equation}
with a $\mu\in\{L,R\}$, and matrices in Eqs.~(\ref{eq:css-blocks}),
(\ref{eq:mat-HXi}), and
\begin{equation}
  \label{eq:Hzpp-matrix}
  H_Z^{(L)}=\left(\begin{array}[c]{cc}B&I-F_A\\ -A&0 \end{array}\right)^T\!\!\!,\;\,\,
  H_Z^{(R)}=\left(\begin{array}[c]{cc}B&0\\ -A&I-F_B \end{array}\right)^T\!\!\!.
\end{equation}
Even though Eq.~(\ref{eq:dZ-quantum-identity-A}) relates the distance
of the original quantum code to those of two other quantum codes with
the same block size, they may still be useful, since the two codes
have half as many basis vectors and exponentially fewer vectors at
large $k$, $q^{k/2}\ll q^k$.

\section{2BGA codes: construction and general properties}
\label{sec:construction}
\subsection{Definition}
2BGA codes are a special case of two-block codes (\ref{eq:css-blocks})
where the commuting matrices are constructed with the help of a group
algebra.  Alternatively, 2BGA codes are a version of GB codes where a
cyclic group is replaced with a general group $G$.  They can also be
thought of as the smallest LP codes.

Given two elements $a,b\in F[G]$ of the group algebra $F[G]$, see
Eq.~(\ref{eq:algebra-element}), with the group size $\ell\equiv |G|$,
the $\ell\times\ell$ matrices $A\equiv \LL(a)$ and
$B\equiv \RR(b)$, respectively, are defined by the left and right
action on group elements,
\begin{equation}
  \label{eq:L-R-action}
  [\LL(a)]_{\alpha,\beta}\equiv \sum_{g\in G}a_g\delta_{\alpha,g\beta},\quad
  [\RR(b)]_{\alpha,\beta}\equiv \sum_{g\in G}b_g\delta_{\alpha,\beta g},
\end{equation}
where group elements $\alpha,\beta\in G$ are used to index rows and
columns, and $\delta_{\alpha,\beta}=1$ if $\alpha=\beta$ and $0$
otherwise is the Kronecker delta.  It is easy to verify that for group
elements $g\in G$, matrices $\LL(g)$ form the regular $F$-linear
representation of $G$.  Further, for any $a,b\in F[G]$,
$\LL(a) \LL(b)=\LL(ab)$, $\RR(a)\RR(b)=\RR(ba)$, while any two
matrices from different sets commute with each
other\cite{Panteleev-Kalachev-2021}, $\LL(a)\RR(b)=\RR(b)\LL(a)$; it
is the latter property that gives the CSS orthogonality condition
(\ref{eq:CSS-orthogonality}).  The map between $\LL(a)$ and $\RR(b)$
can be given in terms of the permutation matrix $P$ with components
$P_{\alpha,\beta}\equiv\delta_{\alpha,\beta^{-1}}$,
$\alpha,\beta\in G$,
\begin{equation}
  \label{eq:L-R-map}
  \LL(a)=P [\RR(a)]^T P,
\end{equation}
where $[\,\cdot\,]^T$ denotes matrix transposition.  The symmetric
permutation operator $P$ acting in $F^\ell$ is equivalent to the map
$\,\bs{\widehat{\ }}:F[G]\to F[G]$ in Eq.~(\ref{eq:hat-operator}),
$P\bs a=\widehat{\bs a}$.  It is also easy to verify that for all
$ a\in F[G]$,
\begin{equation}
  \label{eq:transposition}
  [\LL(a)]^T=\LL(\widehat a),\quad   [\RR(a)]^T=\RR(\widehat a).
\end{equation}

In the following, $\lp[a,b]$ denotes the 2BGA code constructed from
group algebra elements $a,b\in F[G]$, the CSS code
(\ref{eq:css-blocks}) with $A\equiv \LL(a)$ and $B\equiv \RR(b)$ given
by Eq.~(\ref{eq:L-R-action}).  This notation refers to more general LP
codes\cite{Panteleev-Kalachev-2021}, defined in terms of a pair of
matrices with elements in $F[G]$.  Namely, 2BGA codes are a degenerate
case of LP codes with both matrices of dimension $1\times 1$.
Previously considered special cases are GB
codes\cite{Kovalev-Pryadko-Hyperbicycle-2013,%
  Panteleev-Kalachev-2019,Wang-Pryadko-2022}, with $G$ a cyclic group,
and \emph{abelian} 2BGA codes\cite{Kalachev-Panteleev-2020}, with $G$
an abelian group.

\subsection{Code equivalence}
The complexity of enumerating 2BGA codes can be significantly reduced
by excluding permutation-equivalent codes (the proof is given in
Appendix \ref{sec:proofs-IV}):
\begin{theorem}
  \label{th:permutation-equiv}
  For any $a,b\in F[G]$, the 2BGA code $\lp[a,b]$ is equivalent
  to
  \begin{enumerate}[label={\rm (\roman*)}]
  \item \label{theorem3:1}$\lp[\varphi(a),\varphi(b)]$, for any automorphism
    $\varphi:G\to G$;
  \item \label{theorem3:2}$\lp[\alpha^{-1}a\alpha,\beta^{-1}b\beta]$, for any $\alpha,\beta\in G$;
  \item \label{theorem3:3}$\lp[xa,yb]$, for any non-zero $x,y\in F$.
  \item \label{theorem3:4}$\lp[a\alpha ,\beta b]$, for any $\alpha,\beta\in G$;
  \item \label{theorem3:5}$\lp[\widehat b,\widehat a]$;
  \item \label{theorem3:6}In addition, the code  \emph{CSS-dual} to $\lp[a,b]$, with
    interchanged $H_X$ and $H_Z$ matrices, is permutation-equivalent
    to $\lp[\widehat a,-\widehat b]\cong \lp[b,a]$.
  \end{enumerate}
\end{theorem}
Notice that with $\alpha=\beta$, Theorem
\ref{th:permutation-equiv}(ii) is a special case of (i) for inner
automorphisms.  These, and, more generally, item (ii), would be
trivial for an abelian group.  On the other hand, with an abelian
group $G$, for any $a\in F[G]$, $\LL(a)=\RR(a)$, which also gives
$\lp[a,b]\cong \lp[b,a]$, and an immediate consequence, $d_X=d_Z$.  These
properties need not to be true with a non-abelian group $G$.

\subsection{Connectivity of 2BGA codes}
\label{sec:block-structure}

It is convenient to rewrite the CSS equations, e.g., for a non-trivial
codeword ${\bs c}_Z={\bs u\choose \bs v}\in F^{2\ell}$, in terms of
the corresponding pair of group algebra elements, $u,v\in F[G]$.
Direct calculation gives for a 2BGA code $\lp[a,b]$,
\begin{eqnarray}
  \label{eq:cx-orthogonality} au+vb&=&0,\\ %
\left[u+wb, v-aw\right] &\neq& \left[0,0\right],\quad \forall w\in
F[G],
             \label{eq:cx-degeneracy}
\end{eqnarray} where Eq.~(\ref{eq:cx-degeneracy}) enumerates the
degeneracy class.

For a given $a\in F[G]$, consider the subgroup
\begin{equation}
  \label{eq:subgroup-Gx}
  G_a\equiv \left\langle\left\{g\in G:a_g\neq0\right\}\right\rangle,
\end{equation}
the \emph{support group}\cite{Connell-1963} generated by group
elements with non-zero coefficients in $a$,
cf.~Eq.~(\ref{eq:algebra-element}).  Evidently, if we start with any
group element $x\in G$, repeated left multiplication by $a$ can only
generate group algebra elements $ax$, $a^2x$, \ldots, supported on the
left coset $G_ax$ of $x$.  Sizes of left cosets being equal, we have
that the matrix $\LL(a)$ is block-diagonal, with $m_a$ blocks of size
$|G_a|$, where $m_a$ is the index of the support group $G_a$ in $G$,
$m_a \equiv[G:G_a]$.  The same is true for $\RR(b)$, except in this
case we are dealing with the right cosets $xG_b$,
cf.~Eq.~(\ref{eq:cx-orthogonality}), and we may need to order group
elements differently.

Overall, Eq.~(\ref{eq:cx-orthogonality}) implies the block structure
of the code: the row of matrix $H_X$ labeled by the group element
$x\in G$ is in the block associated with the double coset $G_a x G_b$.
Since transposition does not change the support group,
$G_{\widehat a}=G_a$, the same is true for the $x$\,th row of matrix
$H_Z$.  Therefore, if the product of the two subgroups (the double
coset associated with the group identity element $1\in G$) does not
contain all group elements, $G_aG_b\subsetneq G$ (as sets), the code
$\lp[a,b]$ is decomposed into smaller mutually disconnected subcodes
associated with different double cosets in $G_a\backslash G /G_b$.  It
is well known that double cosets do not necessarily have the same
sizes\cite{Bechtell-book-1971}, so the individual double-coset
subcodes are not expected to be equivalent.

To analyze the structure of the matrices in more detail, let us fix an
ordering so that elements of the subgroup $G_a=\{1,g_2,g_3,\ldots\}$
go first in this order, followed by elements of each right coset
$G_ax$, with elements of $G_a$ taken in the same order,
$\{x,g_2x,g_3x,\ldots\}$, and $x\in{\cal A}$, a transversal set of
elements from $G_a\backslash G$ of size $m_a$.  With this choice, it
is easy to see from Eq.~(\ref{eq:L-R-action}) that the $m_a$ blocks of
the matrix $A$ associated with different cosets are identical,
$A=A_1\otimes I_{m_a}$, where $A_1\equiv \LL_{G_a}(a)$, and the
subscript indicates the subgroup that row and column indices are
restricted to.  The same is true for the matrix $B$, except to reveal
the block structure, we may need to take group elements in a different
order.  Denoting the corresponding permutation matrix as $S$, we have
\begin{equation}
  B=S(I_{m_b}\otimes B_1)S^{-1},\quad B_1\equiv \RR_{G_b}(b).
  \label{eq:coset-S-matrix}
\end{equation}
With the ordering of the group $G_a$ fixed, the only remaining freedom
is to order the elements of ${\cal A}$; we can ensure that elements of
each double coset come together, so that decomposition of the 2BGA
code into a direct sum of individual double-coset subcodes be evident.

In general, a (double) coset is not a subgroup of $G$; most of cosets
do not even contain the group identity element.  However, different
double cosets are related to each other by conjugation.  In the case
of support subgroups, we can write $G_axG_b=G_a1G_{xbx^{-1}}$, which
allows to map any double coset to a double coset containing the group
identity element.  The corresponding double-coset subcodes of 2BGA
codes are also related:
\begin{statement}
  \label{th:dbl-coset-codes} A subcode of a disconnected 2BGA code
  $\lp[a,b]$, with some $a,b\in F[G]$, supported in the double coset
  $G_axG_b$, $x\in G$, is equivalent to a subcode of $\lp[a,xbx^{-1}]$
  supported in the double coset $G_a1G_{xbx^{-1}}$.
\end{statement}
In particular, this implies that in the case of an abelian group $G$,
a code equivalent to any double-coset subcode of a 2BGA code
$\lp[a,b]$ over $F[G]$ can be constructed as a 2BGA code over a
subgroup of $G$.  A bit more generally:
\begin{statement}
  \label{th:triple-product}
  If the intersection subgroup $N\equiv G_a\cap G_b$ is abelian and
  normal in both support
  groups, 
  the subcode of $\lp[a,b]$ supported in the double-coset $G_a1G_b$ is
  equivalent to a 2BGA code over a group $G'$ of rank $|G_a 1 G_b|$.
\end{statement}
In particular, with disjoint subgroups, $G_a\cap G_b=\{1\}$, the group
in Statement \ref{th:triple-product} is just a direct product of the
two subgroups, $G'=G_a\times G_b$.  In this case we can independently
choose the order of elements in each subgroup, and both matrices may
simultaneously have the form of Kronecker products,
$A=A_1\otimes I_{n_b}$, $B=I_{n_a}\otimes B_1$, with
$n_b\equiv |G_b|=m_a$ and $n_a\equiv |G_a|=m_b$.  This is exactly the
block structure of an HP code\cite{Tillich-Zemor-2009}, constructed
from square matrices $A_1$ and $B_1$.  If we denote the parameters of
classical linear codes with parity check matrices $A_1$ and $B_1$,
respectively, as $[n_a,k_a,d_a]_q$ and $[n_b,k_b,d_b]_q$ (these
parameters remain the same when the transposed matrices are used), the
parameters of the quantum HP code are known explicitly,
$[[2n_an_b,2k_ak_b,\min(d_a,d_b)]]_q$.  Additional lower and upper
distance bounds on more general codes under conditions of Statement
\ref{th:triple-product} are discussed in Section
\ref{sec:special-central}.

\subsection{Symmetry group of a 2BGA code}

With Eqs.~(\ref{eq:cx-orthogonality}), (\ref{eq:cx-degeneracy}), it is
easy to check the symmetry of a given 2BGA code.  Indeed, for any
$g\in \CC_G(G_a)$, the centralizer of the subgroup $G_a$ in $G$, if a
pair $[u,v]$ is in the code, then the corresponding left-multiplied
pair $[gu,gv]$ is also in the code.  The same is true for the
right-multiplied pairs $[u h,v h]$, $\forall h\in \CC_G(G_b)$.

For an abelian group $G$, we obtain $G$-symmetric analogs of index-two
qQC codes, quasi-abelian 
codes\cite{Panteleev-Kalachev-2020,Kalachev-Panteleev-2020}.  With a
non-abelian $G$, the overall symmetry group of a 2BGA code is
generally smaller than $G$.  In any case, it includes
$\mathop{\rm Z}(G)$, the center of $G$.

\subsection{Code dimension and related codes}

Since 2BGA codes are a subset of general two-block codes, all general
properties from Sec.~\ref{sec:two-block} apply.
The case of GB codes\cite{Kovalev-Pryadko-Hyperbicycle-2013,%
  Panteleev-Kalachev-2019,Wang-Pryadko-2022} is recovered when $G$ is
a cyclic group,
$$
C_\ell\equiv \langle r| r^\ell=1\rangle = \{1,r,r^2,\ldots, r^{\ell-1}\},
$$
where $r^\ell=1$ is implicit in the set notation.  There is an obvious
one-to-one map between the group algebra $F[C_\ell]$ and the ring of
modular polynomials $F[x]/(x^\ell-1)$.  Then, a 2BGA code $\lp[a,b]$
is also a generalized-bicycle code $\gb[a(x),b(x)]$ specified by
polynomials $a(x)$, $b(x)\in F[x]/(x^\ell-1)$, and the square blocks
in Eq.~(\ref{eq:css-blocks}) are just the circulant matrices $A=a(P)$
and $B=b(P)$, where
\begin{equation}
  \label{eq:permutation-matrix}
  P =
  \begin{pmatrix}
    0&\ldots&0 &1\\
    1&&&0\\[-0.5em]
    &\ddots &&\vdots\\
    && 1&0
  \end{pmatrix}
\end{equation}
is an $\ell\times\ell$ cyclic permutation matrix.  A simple expression
for the dimension of a code $\gb[a,b]$ was given in
Ref.~\onlinecite{Panteleev-Kalachev-2019}.  In this case,
$\rank H_X=\rank H_Z=\ell-\deg h(x)$, and
\begin{equation}
  \label{eq:GB-code-size}
  k=2\deg h(x),\quad   h(x)\equiv \gcd\left(a(x),b(x),x^\ell-1\right).
\end{equation}
In fact, $\deg h(x)$ also coincides with the dimension $k_{\rm S}$ of
the quantum cyclic code $\css(A,B^T)$, the single-block subsystem
erasure code (\ref{eq:subsystem-AB}), and Eq.~(\ref{eq:two-block-k})
guarantees that for any cyclic group $G$, $\delta_X=\delta_Z=0$.

More generally, for an abelian group $G$, it is known that
$\rank H_X=\rank H_Z$, and the code dimension is
even\cite{Kalachev-Panteleev-2020}.  An equivalent statement,
$\delta_X=\delta_Z$, also follows from Statement
\ref{th:rank-defect-equal} and Eq.~(\ref{eq:L-R-map}), if we remember
that in the case of an abelian group $G$, for any $a\in F[G]$,
$\LL(a)=\RR(a)$.  A stronger statement can be made whenever a 2BGA
code can be decomposed as a direct sum of GB or HP codes, e.g., under
conditions of Statement \ref{th:triple-product}:
\begin{statement}
  \label{th:abelian}
  Consider a code $\lp[a,b]$ with $a,b\in F[G]$.  If the intersection
  subgroup $N\equiv G_a\cap G_b$ is abelian and normal in both support
  groups, the rank defects of the corresponding CSS matrices vanish,
  $\delta_X=\delta_Z=0$.
\end{statement}
In particular, $\delta_X=\delta_Z=0$
for any abelian 2BGA code.

An alternative sufficient condition follows from Statement
\ref{th:rank-defect-zero} in the special case of a \emph{semisimple}
group algebra $F[G]$, i.e., when the field characteristic $p$ and the
group rank $\ell$ are mutually prime.  Indeed, any ideal in a
semisimple group algebra is a summand, and for any $a\in F[G]$, there
exist idempotent elements $e_a,f_a\in F[G]$ such that $e_a^2=e_a$,
$f_a^2=f_a$, and $e_aa=a$, $a f_a=a$.  In this case, we can choose
$E_A=\LL(e_a)$, $F_A=\LL(f_a)$, which are guaranteed to commute with
$B\equiv \RR(b)$.  A bit of thought gives a more general sufficient
condition:
\begin{statement}
  \label{th:semisimple}
  Consider a code $\lp[a,b]$ over group algebra $R\equiv F[G]$ such
  that the ideals $a R$ and $R a$ (or the two ideals
  generated by $b$) be semisimple.  Then, rank defects of the
  corresponding CSS matrices vanish, $\delta_X=\delta_Z=0$.
\end{statement}
Somewhat less general but easier to apply condition is, e.g., that the
group algebra $F[G_a]$ be semisimple, i.e., rank of the support group
$G_a$ be mutually prime with the field characteristic $p$.

The \emph{semi-abelian} 2BGA codes whose CSS generator matrices have
the property $\delta_X=\delta_Z=0$ are special: their codeword basis
can be chosen so that each codeword is supported on only one block,
similarly to GB codes\cite{Kovalev-Pryadko-Hyperbicycle-2013,%
  Panteleev-Kalachev-2019,Wang-Pryadko-2022} and HP
codes\cite{Tillich-Zemor-2009,Zeng-Pryadko-2018,%
  Zeng-Pryadko-hprod-2020}.  In particular, this gives a lower
distance bound (\ref{eq:d-lower-puncturing}) in terms of the
single-block erasure code, and guarantees the condition of Statement
\ref{th:upper-d-block-diagonal}, giving a simple upper bound on the
distance in terms of matrix block sizes, $d\le\min(|G_a|,|G_b|)$.

However, not all 2BGA codes have this property.  In particular, there
exist \emph{essentially non-abelian} 2BGA codes where
$\delta_X\neq \delta_Z$ or $\delta_X=\delta_Z\neq 0$.
\begin{example}
  Consider the alternating group $A_4$, also known as the rotation
  group $T$ of a regular tetrahedron,
  $$T=\langle x,y|x^3=(yx)^3=y^2=1\rangle,\quad |T|=12,$$ and the
  binary algebra $\mathbb{F}_2[T]$.  Select $a=1+x+y+x^{-1}yx$
  and $b=1+x+y+yx$ to get a 2BGA code
  $\lp[a,b]$ with parameters $[[24,5,3]]_2$.
\end{example}

\subsection{The case of quasi-abelian lifted-product codes}
\label{sec:special-central}

Here we consider in more detail the special case of codes in
Statements \ref{th:triple-product} and \ref{th:abelian}, 2BGA codes
$\lp[a,b]$, with $a,b\in F[G]$ such that the support groups $G_a$ and
$G_b$ have abelian intersection group $N\equiv G_a\cap G_b$ normal
both in $G_a$ and $G_b$.  As discussed, such codes can be seen as
$F$-linear \emph{quasi-abelian} LP codes, or, equivalently, as HP
codes over the abelian group algebra $F[N]$.  Their structure and
parameters can be analyzed using the techniques specific to such
codes.

The following lower and upper bounds are constructed by analogy with
the corresponding theorems from
Ref.~\onlinecite{Kovalev-Pryadko-Hyperbicycle-2013}:

\begin{statement}[Version of Theorem 5 from
  Ref.~\onlinecite{Kovalev-Pryadko-Hyperbicycle-2013}]
  \label{th:lower-d-central-intersection}
  Given elements $a,b\in F[G]$ such that the intersection subgroup
  $N \equiv G_a \cap G_b$ of rank $c$ is abelian and normal in both support
  groups, let $d_A^\perp$ and $d_B^\perp$ be the distances of
  classical $F$-linear group algebra codes with parity check matrices
  $A=\LL(a)$ and $B=\RR(b)$.  Then the distance $d_Z$ of the code
  $\lp[a,b]$ satisfies
  $d_Z\ge d_0\equiv \lceil \min(d^\perp_A,d^\perp_B)/c\rceil$.
\end{statement}

To get a matching upper bound, we need an additional condition to
ensure that, e.g., vectors in $C_A^\perp$ have vectors matching by
symmetry in $C_B^\perp$ to form non-trivial GB codes, see
Eq.~(\ref{eq:gb-decomposition}) in the proof of Statement
\ref{th:abelian}.  It is implicit in the decomposition
(\ref{eq:gb-decomposition}) that we can characterize the symmetry by
ideals of $F[N]$.

Let $J$ be a \emph{maximal} ideal in $F[N]$, and
$\overline {J}= \langle G_a\,{I}\, G_b\rangle$ its extension to the
subspace associated with the coset $G_a1G_b$.  Namely, if
$\mathcal{A}$ and $\mathcal{B}$, respectively, are transversal sets of
representatives from $G_a/N$ and $N\backslash G_b$, every element
$x\in \overline{J}$ can be uniquely written as
\begin{equation}
x=\sum_{\alpha\in\cal A}\sum_{\beta\in\cal B}
\alpha\, x_{\alpha,\beta}\,\beta,\quad x_{\alpha,\beta}\in {J}.\label{eq:dbl-coset-decomposition}
\end{equation}
The corresponding code $C_{J}$, the {\em two-sided coset code\/}
generated by ${J}$, is simply a set of vectors in $F^\ell$
corresponding to elements of $\overline{J}$.  The proof of the
following upper bound is based on the fact that the product of any two
non-zero elements in a maximum ideal is non-zero.
\begin{statement}[Version of Theorem 6 from
  Ref.~\onlinecite{Kovalev-Pryadko-Hyperbicycle-2013}]
  \label{th:upper-d-central-intersection}
  Let ${J}$ be a maximal ideal in $F[N]$, $C_{J}$ the
  two-sided coset code generated by ${J}$, and
  $\widehat{C}_{J}\equiv P C_{J}$ its image under the linear
  map (\ref{eq:hat-operator}).  Denote $d'$ the distance of the
  subcode $C_A^\perp\cap C_{J}$.  Then, if
  $C_{B^T}^\perp \cap \widehat{C}_{J}\neq \{0\}$, the distance of
  the 2BGA code $\lp[a,b]$ satisfies the upper bound, $d_Z\le d'$.
\end{statement}

Evidently, there is also an upper bound in terms of the distance of
the subcode $C_B^\perp\cap C_{J}$.

If we denote the indices of $N$ in the two support groups as
$\ell_a\equiv [G_a:N]$ and $\ell_b\equiv [G_b:N]$, as discussed in
Sec.~\ref{sec:block-structure}, matrices $A$ and $B$ have blocks of
size $c\ell_a$ and $c\ell_b$, respectively.  Then, for a non-trivial
2BGA code, the parameter $d_0$ in Statement
\ref{th:lower-d-central-intersection} satisfies
$d_0\le \min(\ell_a,\ell_b)$, while the upper bounds guarantee
$d_Z\le c\min(\ell_a,\ell_b)$, as would also be expected from
Statement \ref{th:upper-d-block-diagonal}.

The upper and the lower bounds on $d_Z$ coincide when $c=1$: in this
case the subgroup $N=\{1\}$ is trivial so that $F[N]$ is just the
field $F$, and the auxiliary codes in statements
\ref{th:lower-d-central-intersection} and
\ref{th:upper-d-central-intersection} coincide, which gives
$d_Z=\min(d_A^\perp,d_B^\perp)$.  Of course, the same result for the
distance can be also obtained from the map to a hypergraph-product
code constructed from the single-block classical group algebra codes
with groups $G_a$ and $G_b$.

It is known\cite{Borello-delaCruz-Willems-2022} that classical group
algebra codes include good codes with finite rates and finite relative
distances.  This guarantees the existence of finite-rate 2BGA codes
with distance scaling as a square root of block length.

\subsection{2BGA codes with row weights $W\le 4$.}
\label{sec:special-W4}

Let us discuss 2BGA codes with small row weights
$W\equiv \wgt(a)+\wgt(b)$.  Evidently, a code $\lp[a,b]$ with, e.g.,
$a=0$ has one block zero, which immediately gives $d_X=d_Z=1$, as long
as the code is non-trivial.  Similarly, according to Theorem
\ref{th:permutation-equiv}, any group algebra element with $\wgt(a)=1$
can be equivalently replaced with $a=1$, giving $A=I_\ell$, the
identity matrix, which gives $\rank H_X=\rank H_Z=\ell$, and thus a
trivial code with $k=0$.  Thus, to get a useful code with row weight
$W\le 4$, we must have $\wgt(a)=\wgt(b)=2$.

In this case, up to code equivalence, we can choose $a=1+\lambda f$,
$b=1+\mu h$, with non-zero $\lambda,\mu\in F$ and non-identity group
elements $f,h\in G$, so that the support groups $G_a=\langle f\rangle$
and $G_b=\langle h\rangle$ be cyclic.  Thus, according to Statement
\ref{th:triple-product}, all non-trivial 2BGA codes with row weight
$W=4$ are equivalent to direct sums of abelian-group codes.  Similarly
to GB codes with row weight $4$, these codes can be mapped to rotated
surface codes\cite{Wang-Pryadko-2022}.

Indeed, given a group element $g_{0,0}\equiv g\in G$, the double coset
$G_a gG_b$ can be parameterized as $g_{x,y}=f^xg h^y$, with the
positions $(x,y)\in \mathbb{Z}^2$ on the integer plane corresponding
to the same group elements identified.  In particular,
$(x,y)\simeq (x+\ord f,y)\simeq (x,y+\ord h)$, where $\ord f$ is the
order of the group element $f$.  Positive displacements along
horizontal and vertical edges correspond to multiplication by
$\lambda g$ and $\mu h$, respectively.  This way, we obtain a finite
locally planar graph covered by the infinite square lattice.  It is
easy to check that the local structure of the CSS code associated with
the double coset $G_a gG_b$ is exactly that of a square-lattice
surface code, so that the corresponding codewords are homologically
non-trivial chains (or co-chains) connecting pairs of identified
vertices (faces) on the integer plane.

The nature of the resulting double-coset subcodes depends on the
homology group associated with the covering map
$\varphi:\mathbb{Z}^2\to \{g_{x,y}|x,y\in \mathbb{Z}\}$.  For example, with
$g=1$, we get a toric code if
$\langle f\rangle\cap \langle h\rangle=\{1\}$.  More generally, we get
a surface code on a finite transitive graph which locally looks like a
square lattice.  The hand-shaking lemma guarantees that every
connected component with $V$ vertices has $2V$ edges, and its dual
version gives $V$ faces, which gives a $k=2$ surface code for any
connected component.

Further, a counting argument identical to that used in the proof of
Statement 14 from Ref.~\onlinecite{Wang-Pryadko-2022} gives an upper
bound for the $W=4$ double-coset subcode distances $d_X^{(g)}$ and
$d_Z^{(g)}$ in terms of its length $n^{(g)}=2\left|G_a g G_b\right|$,
\begin{equation}
  \bigl(d_\mu^{(g)}\bigr)^2\le n^{(g)},\quad \mu\in\{X,Z\},\label{eq:w4-upper-bound}
\end{equation}
which also gives $d^2\le n^{(g)}-1$ when $d\equiv d^{(g)}$ is odd.
For both inequalities, we found many cases of saturation.

\section{Numerical results}
\label{sec:numerical}

In this section, we present optimal parameters of short qubit-based
connected 2BGA codes found numerically.  Namely, we computed the
parameters of all inequivalent non-trivial connected binary 2BGA codes
$\lp[a,b]$ with $\wgt(a)+\wgt(b)\le 8$, for all non-abelian groups $G$
of orders $\ell\le 100$, and all abelian groups of orders $\ell\le 50$,
for each group keeping only the first found code with a given
dimension $k$ and distance $d$.

We should point out that the double-coset subcodes could, potentially,
have better parameters than connected 2BGA codes of equal size.
Nevertheless, for a given block length $n$, we do not have a way to
limit the group sizes which would result in subcodes of length $n$.
Therefore, to speed up the calculation, we decided to only consider
the connected codes.

We used the Small Groups library distributed with GAP\cite{GAP4} to enumerate
inequivalent groups, and the GAP package {\tt QDistRnd}
\cite{Pryadko-Shabashov-Kozin-QDistRnd-2021} to calculate the code
distances.  The calculations were performed at the UCR High
Performance Computing Center.

Specifically, to eliminate permutation-equivalent codes for a given
group $G$, we used the default order of group elements in GAP to
establish the alphabetical order of subsets of $G$, which also gives
an ordering for elements of $\mathbb{F}_2[G]$.  In the following, we
write $a<b$ if $a$ goes before $b$ in this order.  Given the weights
$W_a$ and $W_b$, for each pair $(a,b)$ generated consecutively with
$\wgt(a)=W_a$, $\wgt(b)=W_b$, and also $a<b$ if $W_a=W_b$, we
discarded all pairs where $\alpha a\beta< a$ or $\alpha b\beta<b$ for
any $\alpha,\beta\in G$.  Indeed, according to Theorem
\ref{th:permutation-equiv}, these inequalities indicate that a
permutation-equivalent code has already been encountered.  Since the
identity group element $1$ is always the first in the list, we only
needed to consider group algebra elements with $a_1=b_1=1$.  In
addition, with $W_a=W_b$, we discarded the pairs with $\widehat b< a$,
see Theorem \ref{th:permutation-equiv} \ref{theorem3:5}.  After
constructing the matrices, we also made sure to drop all disconnected
codes.

We note that for \emph{even-even} codes, with both $W_a$ and $W_b$
even, rows and columns of the matrices (\ref{eq:L-R-action}) have even
numbers of elements.  For the binary field $\mathbb{F}_2$, this
guarantees that the rows of the matrices $A$, $B$, and thus of the
stabilizer generator matrices $H_X$, $H_Z$ add to zero;
Eq.~(\ref{eq:k-CSS}) guarantees that we get non-trivial codes with
$k\ge 2$.  In comparison, with one or both weights odd, there are many
trivial codes with $k=0$.

\begin{figure}[htbp]
    \centering
    \includegraphics[width=0.49\textwidth]{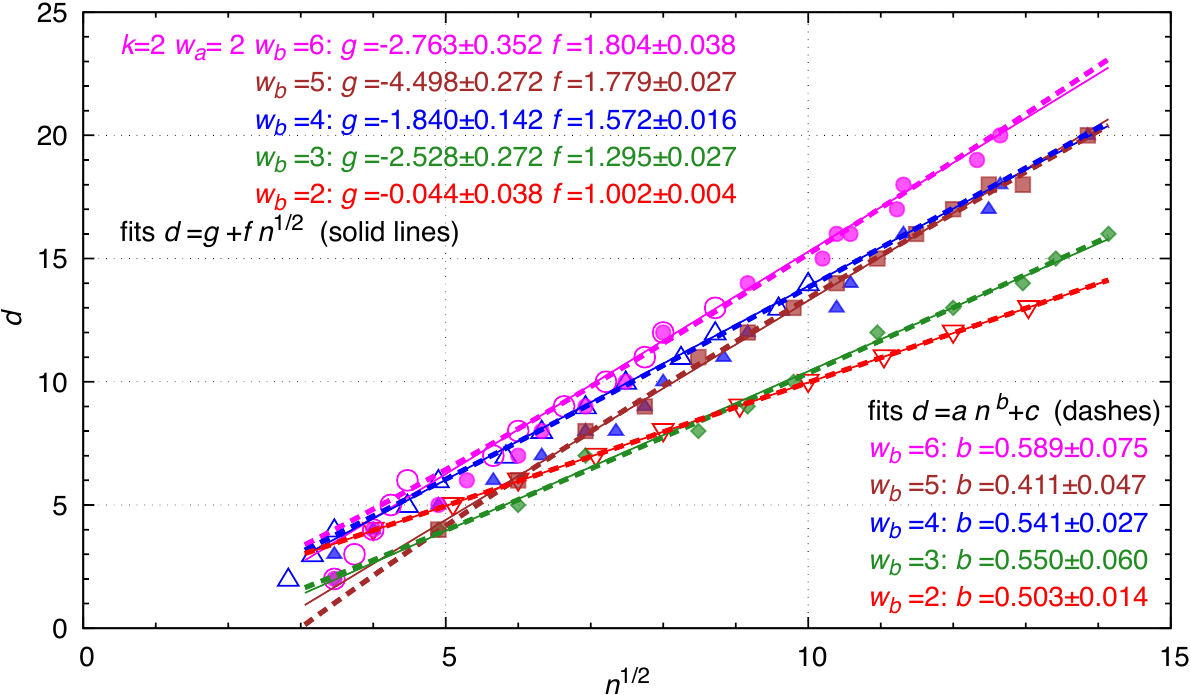}
    \caption{(color online) Distance $d$ of connected 2BGA codes
      encoding $k=2$ qubits, with weights $W_a=2$ and $W_b$ as
      indicated, plotted as a function of the square root of the block
      size $n$.  Red upside-down triangles $\triangledown$, green
      diamonds {\large$\diamond$}, blue triangles $\triangle$, brown
      squares $\square$, and purple circles
      {\scriptsize\raisebox{0.2em}{$\bigcirc$}}, respectively,
      correspond to total weights $W=4$, $5$, $6$, $7$, and $8$.  Open
      symbols correspond to abelian groups, filled symbols to
      non-abelian groups.  Only the shortest codes with each $k$ and
      $d$ found are shown.  The solid and dashed lines, respectively,
      are the fits using $d=g+f n^{1/2}$ and $d=a n^b+c$; the
      coefficients are listed in the captions.}
    \label{fig:disk2}
\end{figure}

\begin{figure}[htbp]
    \centering
    \includegraphics[width=0.49\textwidth]{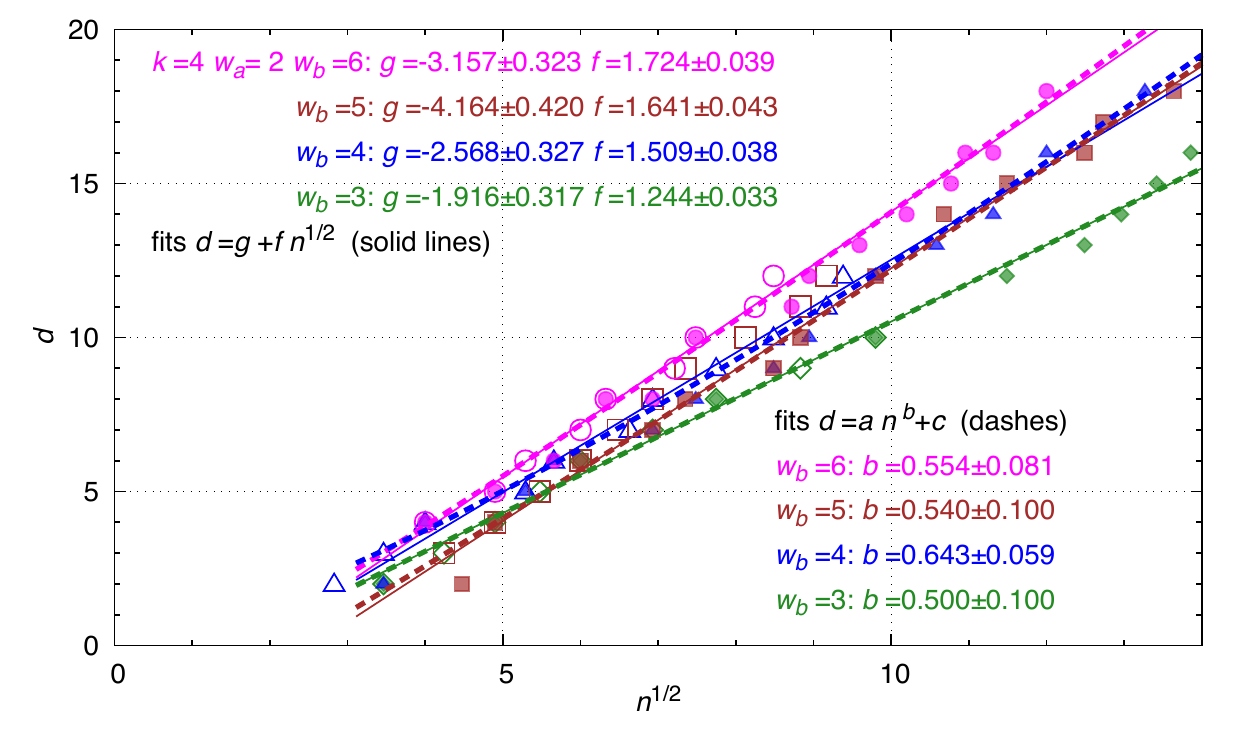}
    \caption{(color online) As in Fig.~\ref{fig:disk2} but for
      codes encoding $k=4$ qubits.  In agreement with the results in
      Sec.~\ref{sec:special-W4}, there are no connected codes with
      $k>2$ and $W_a=W_b=2$.}
    \label{fig:disk4}
  \end{figure}

 \begin{figure}[htbp]
    \centering
    \includegraphics[width=0.49\textwidth]{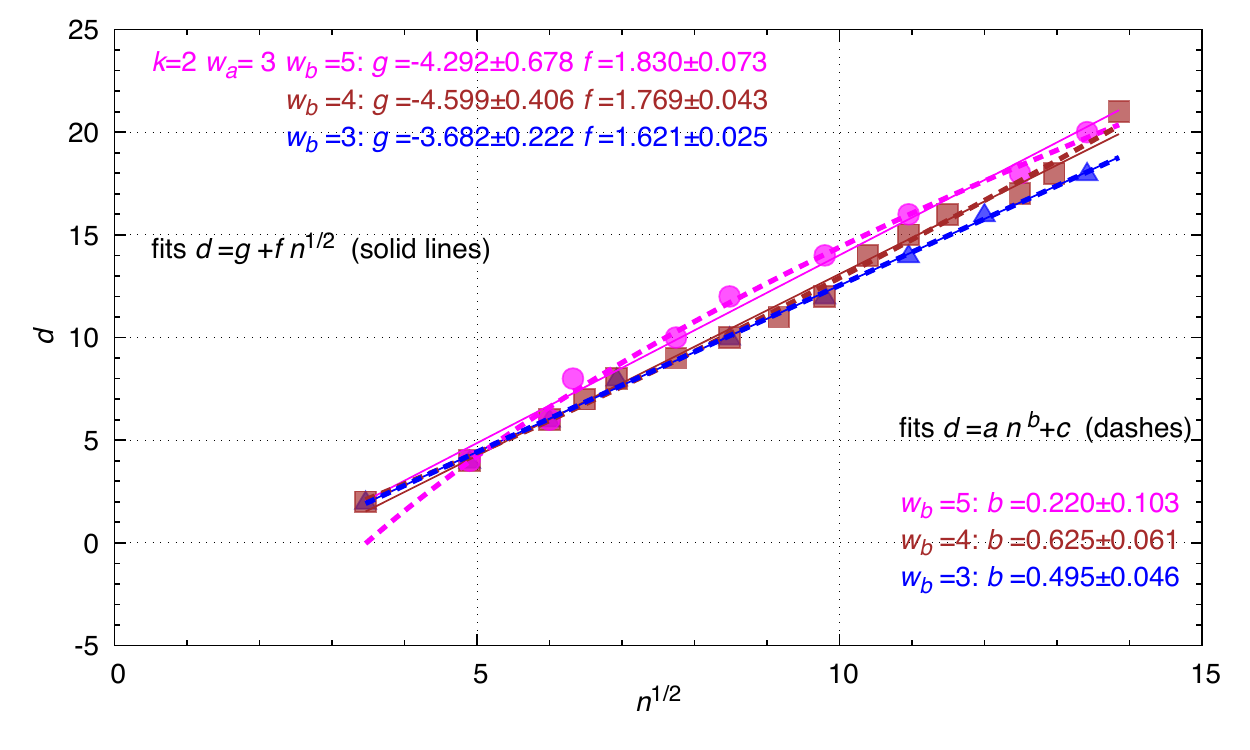}
    \caption{(color online) As in Fig.~\ref{fig:disk2} but for 2BGA
      codes with weights $W_a=3$ and $W_b$ as indicated.  We have not
      found any abelian-group codes with these weights and $k=2$.}
      \label{fig:disk2wa3}
\end{figure}

\begin{figure}[htbp] \centering
  \includegraphics[width=0.49\textwidth]{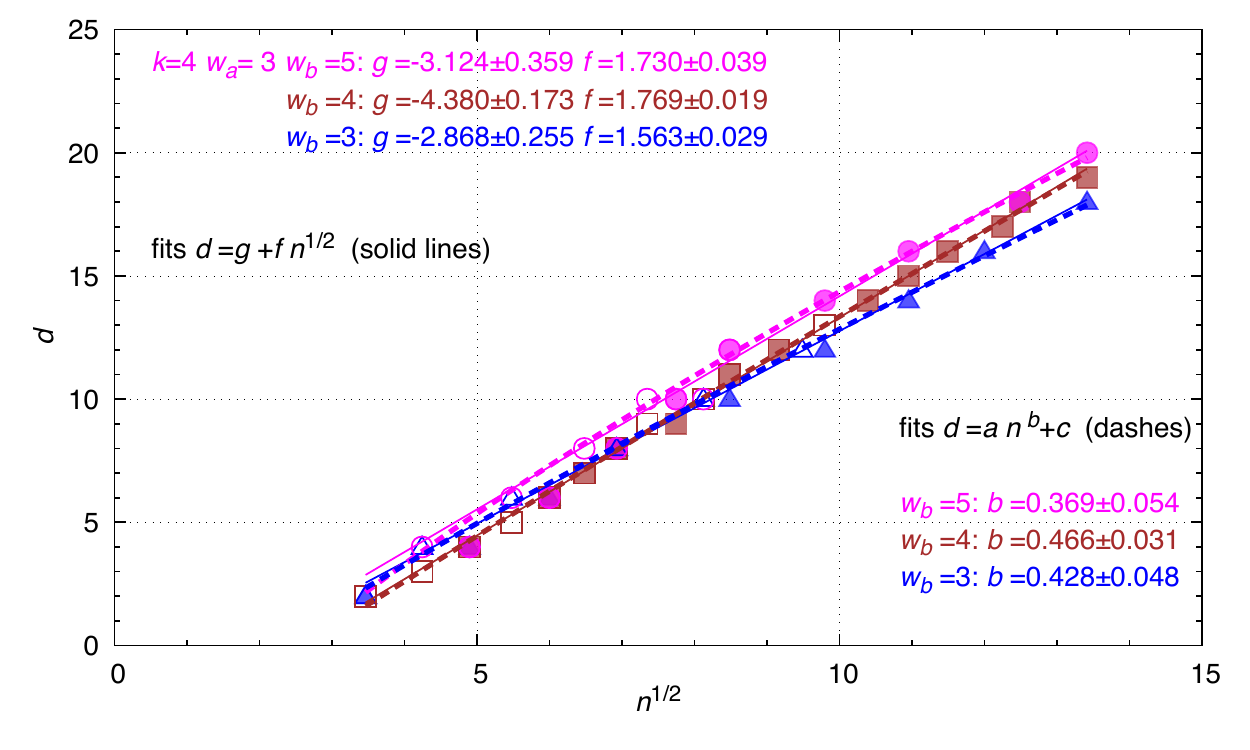}
  \caption{(color online) As in Fig.~\ref{fig:disk2wa3} but for 2BGA
    codes encoding $k=4$ qubits.}
    \label{fig:disk4wa3}
\end{figure}

\begin{figure}[htbp] \centering
  \includegraphics[width=0.49\textwidth]{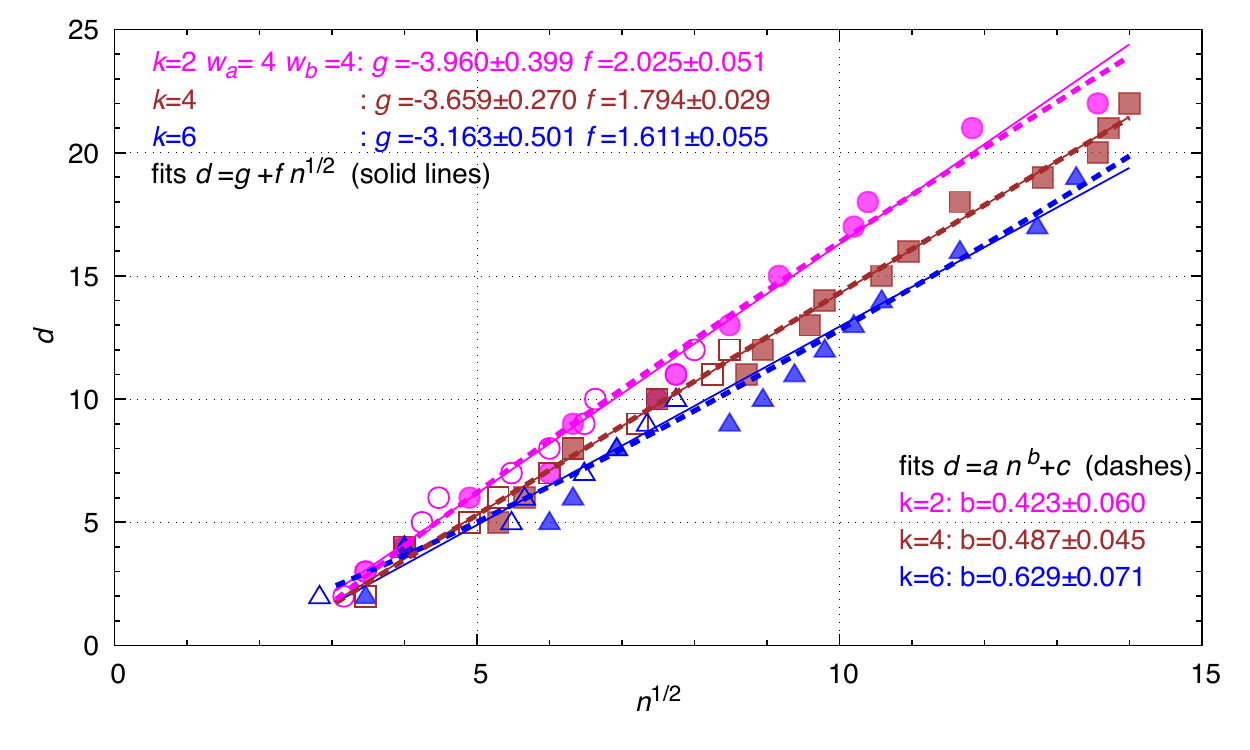}
  \caption{(color online) As in Fig.~\ref{fig:disk2} but for 2BGA
    codes with weights $W_a=W_b=4$ and encoding $k$ qubits as
    indicated.  Circles {\scriptsize\raisebox{0.2em}{$\bigcirc$}},
    squares $\square$, and triangles $\triangle$, respectively,
    correspond to $k=2$, $4$, and $6$.}
  \label{fig:disk2_4_6}
\end{figure}

The computed distances $d$ for codes with $W_a=2$ are plotted as a
function of the square root of the block size $n$ in
Figs.~\ref{fig:disk2} and \ref{fig:disk4} for codes with $k=2$ and
$k=4$, respectively.  Different symbols and colors correspond to row
weights $W\in \{ 4,5,6,7,8\}$ as indicated in the caption of
Fig.~\ref{fig:disk2}, with open and closed symbols corresponding to
codes obtained from abelian and non-abelian groups, respectively.
Figs.~\ref{fig:disk2wa3} and Figs.~\ref{fig:disk4wa3} give similar
data for codes with $W_a=3$ and $W_b$ as indicated, and
Fig.~\ref{fig:disk2_4_6} shows distances for small-$k$ codes with
$W_a=W_b=4$.  These plots all look similar to a family of GB codes
with $k=2$ studied in Ref.~\onlinecite{Wang-Pryadko-2022}.  Namely,
the available largest distances show reasonable agreement with
asymptotic distance scaling $d=g+f n^{1/2}$, with the slope
$f\equiv f(W_a,W_b)$ an increasing function of the total row weight
$W=W_a+W_b$, while different values of $W_a\ge 2 $ and $W_b\ge W_a$
have a relatively minor effect on the coefficients $g$ and $f$.

Square-root scaling of the distance is in agreement with the lower
bound in Statement \ref{th:lower-d-central-intersection} and, in
particular, with the case $G_a\cap G_b=\{1\}$, where 2BGA codes can be
represented as hypergraph-product codes constructed from a pair of
classical group-algebra codes.  To compare, we also tried fitting the
distances with $d=a n^b+c$, where $a$, $b$, and $c$ are parameters.
While an upward or downward curvature corresponding to exponent
$b>1/2$ or $b<1/2$, respectively, can be seen on some of the fits, the
actual deviations from the linear (in $n^{1/2}$) fits are small,
$\Delta d\alt 0.2$ on most plots.  We conclude that 2BGA codes with
dimensions $k\le 4$, row weights $W\le8$, and group sizes studied so far
are not nearly large enough to resolve the question about the scaling
of the code distances of such codes with the block size.

While constructed 2BGA codes with $k>4$ also have maximum distances
$d$ scaling near linearly with $n^{1/2}$, the corresponding
coefficients show relatively little dependence on $k$ (data not
shown).  For this reason, and to reveal the patterns in code
parameters, in Figs.~\ref{fig:k_nwt2+6A} and \ref{fig:k_nwt2+6}, we
plot the distances $d$ of the found 2BGA codes with $k\ge4$ as a
function of $n$, for codes with $W_a=2$ and $W_b=6$ obtained from
abelian and non-abelian groups, respectively.  Most prominent in
Fig.~\ref{fig:k_nwt2+6A} are the sequences of codes with $kd=n$ with
$k\ge 6$ and the distances $d=\{2,3,\ldots ,d_{\rm max}(k,n)\}$, where
the sequence cut-off $d_{\rm max}(k,n)=\mathcal{O}(n^{1/2})$ shows
relatively little dependence on $k$.  Codes with $W_a=2$, $W_b=6$
obtained from non-abelian groups (Fig.~\ref{fig:k_nwt2+6}) also have
sequences with $kd=n$, but only for $k=6$ or doubly-even
$k\in \{4,8,12,16,\ldots\}$; there are also sequences of codes with
$k=4s+2$, $s\ge2$, whose parameters satisfy the relation $(k+2)d=n$.
As we demonstrate in Appendix \ref{sec:examples}, all codes in the
sequences with $kd=n$ can be obtained starting with groups $C_{mh}$
and $D_{m}$, both of which give index-$4$ qQC codes of length
$n=4m$.

\begin{figure}[htbp]
  \centering \includegraphics[width=0.49\textwidth]{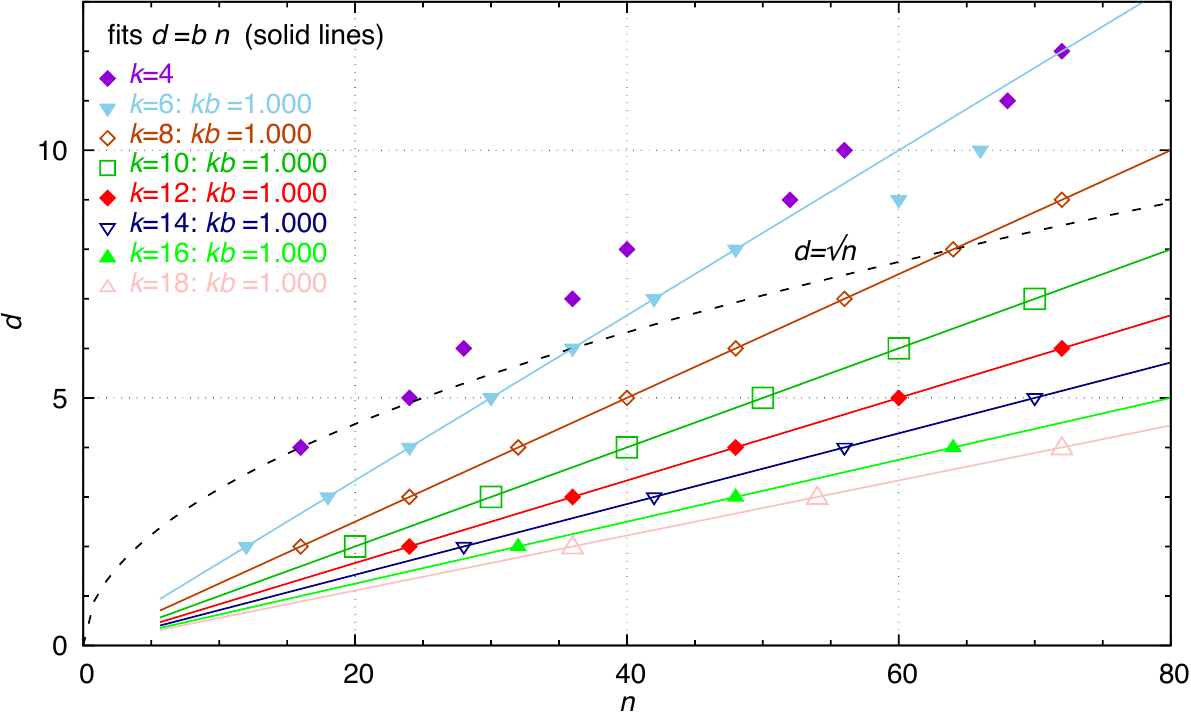}
  \caption{(color online) Distances $d$ of abelian connected 2BGA
    codes with $W_a=2$, $W_b=6$, and $k\ge 4$, plotted as a function
    of the block length $n$.  Different symbols correspond to actual
    codes found, with $k$ values as indicated in the caption.  Solid
    lines are fits to $d=bn$ using only the data on or below the black
    dashed line, $d= n^{1/2}$.  Parameters of most of the codes with
    $k\ge6$ satisfy the relation $kd=n$ exactly.  This can be seen
    from the values of the product $kb$ shown in the caption; the
    sequences terminate at $d=\mathcal{O}(n^{1/2})$ more or less
    independent of the value of $k$.  All the codes fitting this
    pattern can be obtained from the groups $C_{mh}=C_m\times C_2$,
    $m\ge3$, of order $\ell=2m$, producing index-4 qQC CSS codes of
    length $n=4m$.}
  \label{fig:k_nwt2+6A}
\end{figure}

\begin{figure}[htbp]
  \centering
  \includegraphics[width=0.49\textwidth]{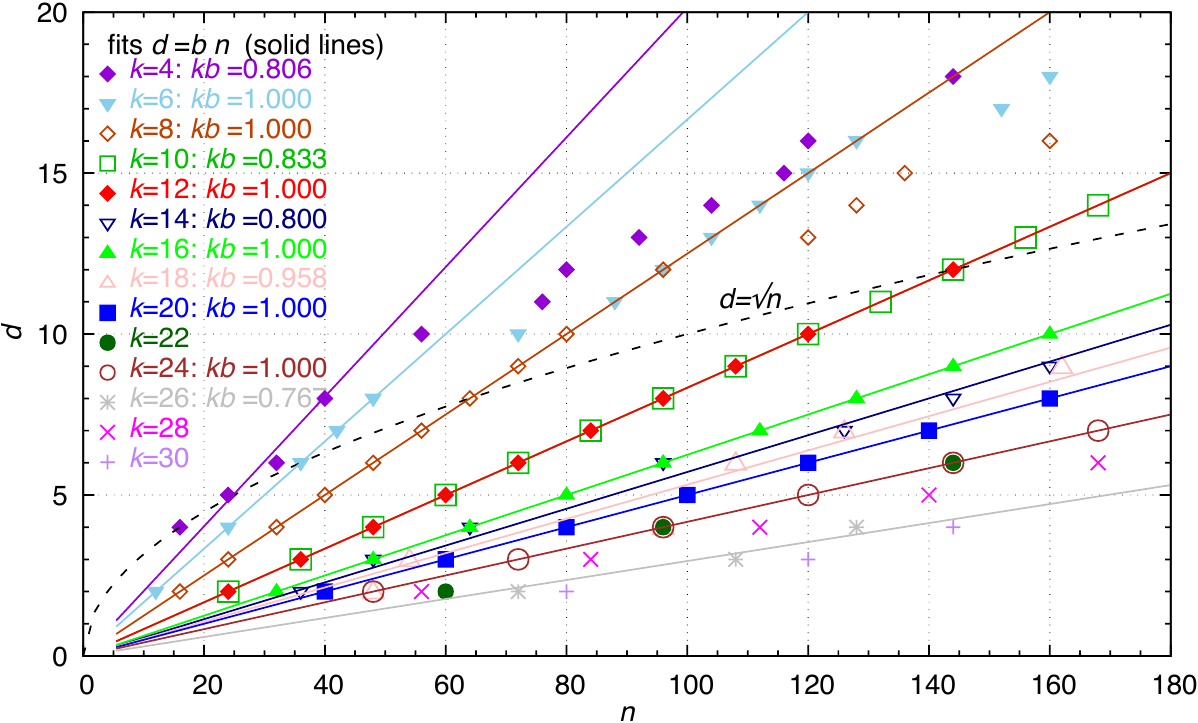}
  \caption{(color online) As in Fig.~\ref{fig:k_nwt2+6A} but for
    connected 2BGA codes with $W_a=2$, $W_b=6$, and $k\ge 4$ even,
    obtained from non-abelian groups.  As can be seen from
    Tab.~\ref{tab:D2m-2+6} in Appendix~\ref{sec:examples}, sequences
    of codes with $kd=n$ and doubly-even $k=4s$, $s\ge1$, can be
    obtained from the groups $D_{m}$ which give index-4 qQC codes of
    length $n=4m$; the sequences terminate at
    $d=\mathcal{O}(n^{1/2})$.}
  \label{fig:k_nwt2+6}
\end{figure}

Fig.~\ref{fig:kd_nwt2+6} shows the same data as
Fig.~\ref{fig:k_nwt2+6} but with $kd$ plotted as a function of $n$.
It demonstrates that for all constructed codes with $W_a=2$ we have
$kd\le n$.  By this measure, the best among codes with $W_a=2$ are
index-4 qQC codes from groups $C_{mh}$ and $D_{m}$ with $kd=n$.  We
should note that, according to Example 12 in
Ref.~\onlinecite{Kovalev-Dumer-Pryadko-2011}, sufficiently long codes
with $kd=n$ can be obtained as additive cyclic codeword-stabilized
(CWS) codes\cite{Chuang-CWS-2009,Cross-CWS-2009,Chen-Zeng-Chuang-2008}
from a set of $k$ classical repetition codes, although such codes do
not necessarily have bounded stabilizer weights.  Of course,
asymptotically the ratio $kd/n$ increases without a bound for many
families of quantum LDPC codes, e.g., as $\mathcal{O}(n^{1/2})$ for
hypergraph-product codes\cite{Tillich-Zemor-2009}. However, it is not
trivial to get $kd> n$ in a degenerate quantum code of length
$n\lesssim 10^2$.  Some of such 2BGA codes constructed here are listed
in Table~\ref{tab:large-k}.

\begin{figure}[htbp]
  \centering \includegraphics[width=0.49\textwidth]{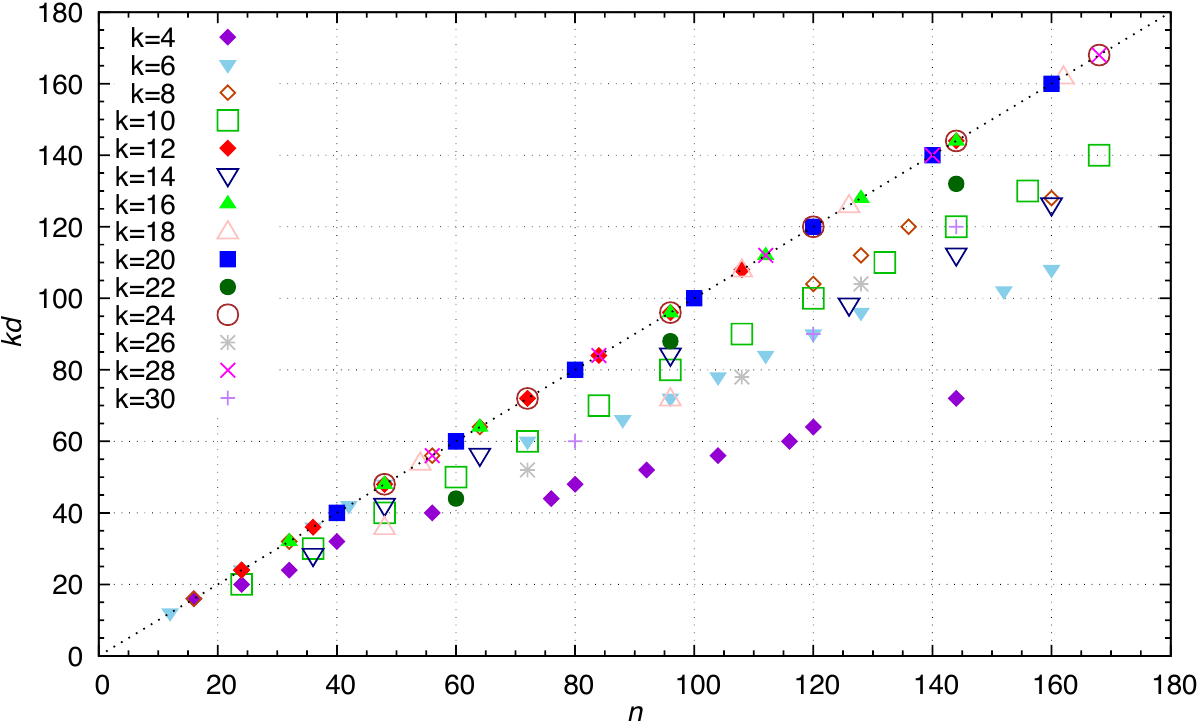}
  \caption{(color online) Same data as in Fig.~\ref{fig:k_nwt2+6}, but
    with the products $kd$ plotted as a function of $n$.  The dotted
    line is the diagonal, $kd=n$; all of the found codes with $W_a=2$
    and $W_b=6$ have $kd\le n$.}
    \label{fig:kd_nwt2+6}
\end{figure}

\begin{table*}[htbp]
  \centering
  \begin{tabular}[c]{c|c||c|c|c||c|c|c|c}
    $\ell$& $\#$&$n$&$k$&$d$&$a$&$b$&presentation&structure\\ \hline\hline 
    36 & 2 & 72 & 8 &  9   &$1+ r^{28} $ &$1+ r^9+ r^{18}+ r^{12}+ r^{29}+ r^{14} $  &$\langle r|r^{36}\rangle $ & $C_{36}$\\ \hline
                                                                                                                                                     
    36 & 1 & 72 &  8 &  9    &$1+ r $ &$1+ s+ r^6+ s^3r+ sr^7+ s^3r^5 $  &$\langle r,s|s^4,r^6,s^{-1}rsr\rangle $& $C_9\ltimes C_4$\\                
    40 & 1 & 80 & 8 &  10   &$1+ sr^4 $ &$1+ r+ r^2+ s+ s^3r+ s^2r^6 $  &$\langle r,s|s^5,r^8,r^{-1}srs\rangle $ & $C_5\ltimes C_8$\\
    48 & 10& 96 & 8 &  12  &$1+ sr^2 $ &$1+ r+ s^3+ s^4+ s^2r^5+ s^4r^6 $  &$\langle r,s|s^6,r^8,(rs)^8\rangle $ & $(C_3\ltimes C_8)\ltimes C_2$\\ \hline 
                                                                                                                                                     
    27 & 1 & 54 & 6 &  9   &$1+ r+ r^3+ r^7 $ &$1+ r+ r^{12}+ r^{19} $  &$\langle r|r^{27}\rangle $ & $C_{27}$\\
    30 & 4 & 60 & 6 &  10 &$1+ r^{10}+ r^6+ r^{13} $ &$1+ r^{25}+ r^{16}+ r^{12} $   &$\langle r|r^{30}\rangle $ & $C_{30}$\\
    35 & 1 & 70 &  8 &  10   &$1+ r^{15}+ r^{16}+ r^{18} $ &$1+ r+ r^{24}+ r^{27} $  &$\langle r|r^{35}\rangle $ & $C_{35}$\\
    36 & 2 & 72 & 8 &  10  &$1+ r^9+ r^{28}+ r^{31} $ &$1+ r+ r^{21}+ r^{34} $  &$\langle r|r^{36}\rangle $ & $C_{36}$\\
    36 & 2 & 72 & 10&  9    &$1+ r^9+ r^{28}+ r^{13} $ &$1+ r+ r^3+ r^{22} $  &$\langle r|r^{36}\rangle $ & $C_{36}$\\ \hline
                                                                                                                                                     
    36 & 1 & 72 & 8 &  9  &$1+ s+ r+ sr^6 $ &$1+ s^2r+ s^2r^6+ r^2 $  &$\langle r,s|s^4,r^9,s^{-1}rsr\rangle $    & $C_9\ltimes C_4$\\
    40 & 1 & 80 & 8 &  10 &$1+ r+ s+ s^3r^5 $ &$1+ r^2+ sr^4+ s^3r^2 $  &$\langle r,s|s^5,r^8,s^{-1}rsr\rangle $  & $C_5\ltimes C_8$\\
    48 & 1 & 96 & 8 &  12 &$1+ r+ s+ r^{14} $ &$1+ r^2+ sr^4+ r^{11} $  &$\langle r,s|s^3,r^{16},r^{-1}srs\rangle$ & $C_3\ltimes C_{16}$\\
    40 & 8 & 80 & 9 &  9  &$1+ sr^5+ r^5+ sr^6 $ &$1+ s^2+ r+ s^2r^3 $  &$\langle r,s|s^4,r^{10},(rs)^2\rangle $  & $(C_{10}\times C_2)\ltimes C_2$\\
    42 & 3 & 82 & 10&  9  &$1+ r^7+ r^8+ sr^{10} $ &$1+ s+ r^5+ s^2r^{13} $  &$\langle r,s|s^3,r^{14},r^{-1}srs\rangle $ & $C_7\times S_3$\\
    48 & 13& 96 & 10&  12 &$1+ s+ r^9+ sr $ &$1+ s^2r^9+ r^7+ r^2 $  &$\langle r,s|s^4,r^{12},s^{-1}rsr\rangle$    & $C_{12}\ltimes C_4$\\
    48 & 5 & 96 & 11&  9  &$1+ s+ r^9+ sr^{13} $ &$1+ r^9+ sr^{18}+ r^7 $  &$\langle r,s|s^2,r^{24},(rs)^8\rangle$ & $C_{24}\ltimes C_2$ \\
    48 & 9 & 96 & 12&  10 &$1+ r+ s^3r^2+ s^2r^3 $ &$1+ r+ s^4r^6+ s^5r^3 $ &$\langle r,s|s^6,r^8,r^{-1}srs\rangle$ & $C_2 \times (C_3\ltimes C_8)$
  \end{tabular}
  \caption{Parameters of degenerate connected 2BGA codes with
    $n< 10^2$, $k d\ge n$, and $d>W$.  Only codes with the largest
    examined row weight $W=8$ have been found with such
    parameters. Here $\ell$ is the group order, ``$\#$'' is the group
    number specific to GAP, $n$, $k$, $d$ are parameters of the code
    $\lp[a,b]$ with $a$ and $b$ as indicated, ``presentation'' is the
    shortest group presentation, and ``structure'' is the output of a
    function call ``{\tt
      StructureDescription(SmallGroup(}$\ell$,$\#${\tt ));}'' in GAP,
    with $C_m$ an order-$m$ cyclic group, $S_3$ the order-$6$
    symmetric group, ``$\times$'' the direct product of groups, and
    ``$\ltimes$'' the semidirect product, with the normal subgroup on
    the left.  The four row blocks, respectively, list codes with
    $W_a=2$, $W_b=6$, first abelian then non-abelian, followed by
    codes with $W_a=W_b=4$, first abelian then non-abelian.  The terms
    in group algebra elements $a$ and $b$ are sorted according to the
    internal presentation in GAP.}
\label{tab:large-k}
\end{table*}

\begin{figure}[htbp] \centering
  \includegraphics[width=0.49\textwidth]{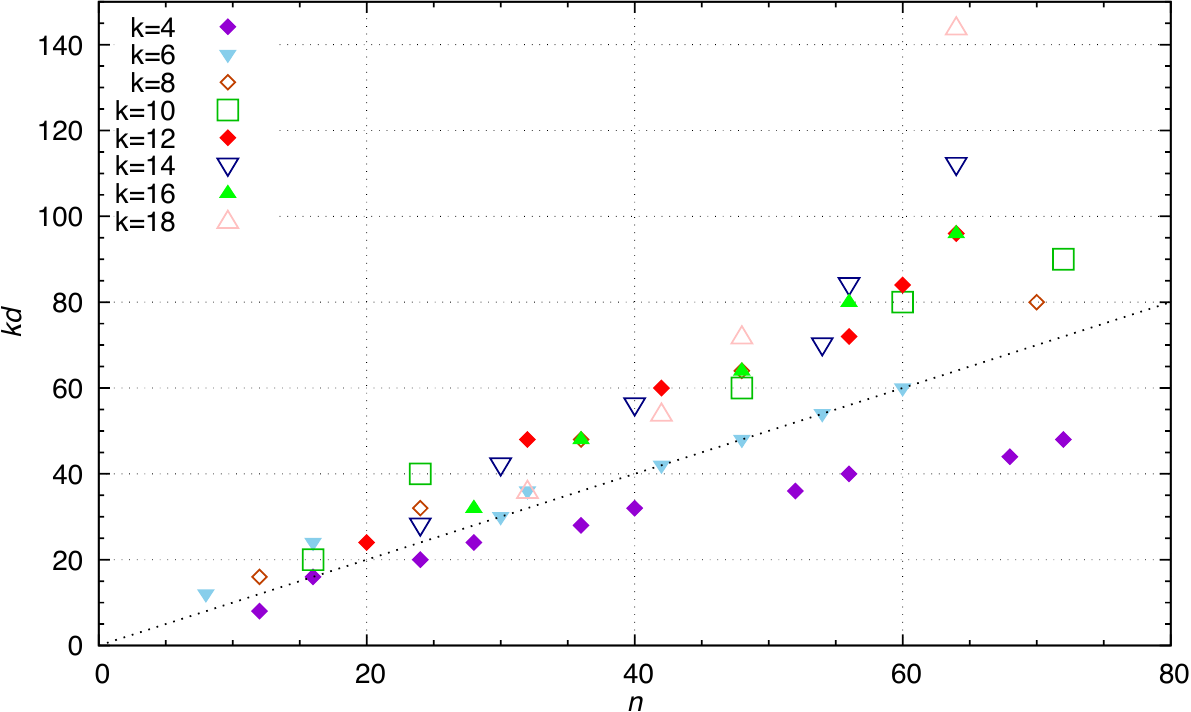}
  \caption{(color online) As on Fig.~\ref{fig:kd_nwt2+6} but for
    abelian connected 2BGA codes with $W_a=W_b=4$ and $k\ge 4$.}
  \label{fig:kd_nwt4+4A}
\end{figure}

\begin{figure}[htbp] \centering
  \includegraphics[width=0.49\textwidth]{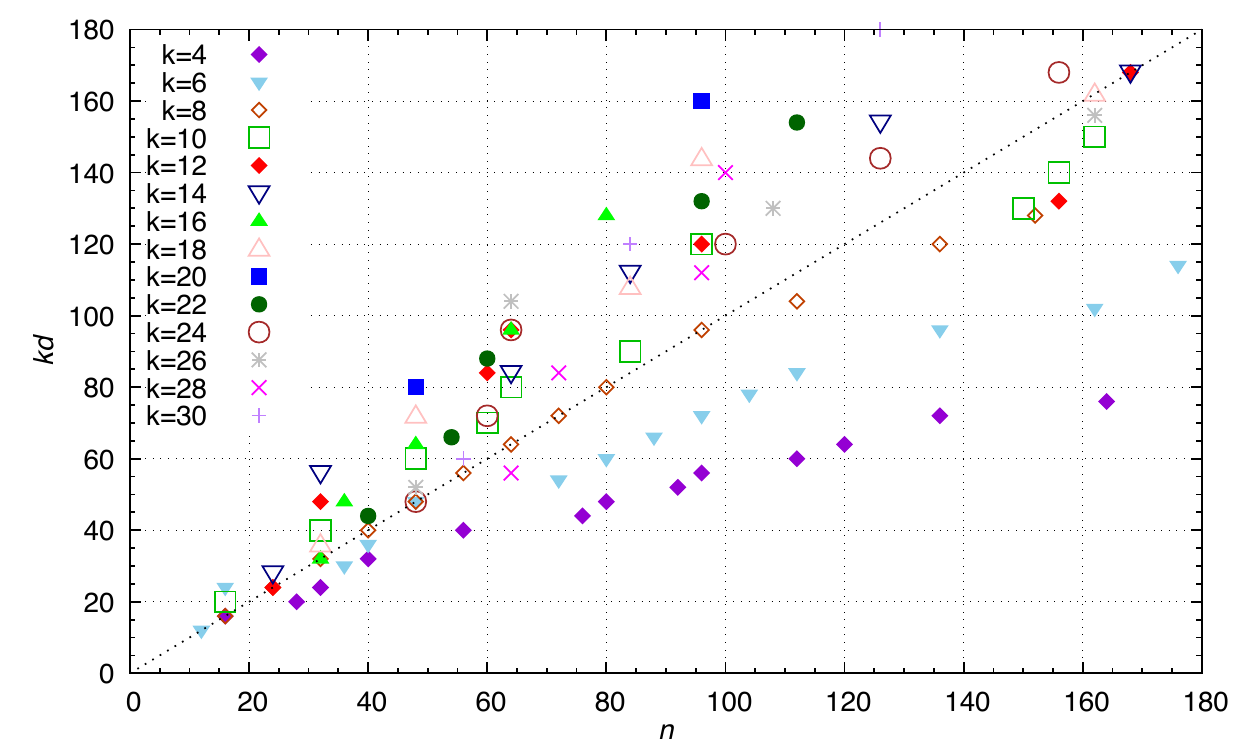}
  \caption{(color online) As on Fig.~\ref{fig:kd_nwt4+4A} but for
    codes with $W_a=W_b=4$ and $k\ge4$ even obtained from non-abelian
    groups.}
  \label{fig:kd_nwt4+4}
\end{figure}

Most of the constructed 2BGA codes with $W_a=W_b=4$ do not have simple
relations between their parameters, and the plots of $d$ vs.\ $n$ are
not illuminating.  For this reason, we decided to illustrate the
parameters of such codes with $k\ge 4$ by plotting $kd$ vs.\ $n$, for
abelian codes in Fig.~\ref{fig:kd_nwt4+4A} and for non-abelian codes
in Fig.~\ref{fig:kd_nwt4+4}.  Only codes with even values of $k$ are
shown.  As evident from the plots, many of the constructed codes have
$kd>n$, including the abelian code $[[64,18,8]]$ with $kd/n>2$
obtained from the group $C_4\times C_4\times C_2$.

\section{Conclusions}
\label{sec:conclusions}

In conclusion, we introduced and studied analytically and numerically
a family of quantum 2BGA codes, an ansatz particularly suitable for
constructing short and intermediate-length quantum LDPC codes.
Indeed, unlike for many of the ``product''
constructions\cite{Tillich-Zemor-2009,%
Hastings-Haah-ODonnell-2020,Panteleev-Kalachev-2020,%
Breuckmann-Eberhardt-2020,Panteleev-Kalachev-2021} which tend to give
very long quantum codes, the block length of a 2BGA code is twice the
size of the group used in the construction.  Further, unlike for
quantum group-algebra codes in
Ref.~\onlinecite{Naghipour-Jafarizadeh-Shahmorad-2015} which are
analogs of quantum cyclic codes, here the CSS orthogonality constraint
is naturally satisfied for any pair of group algebra elements.

Moreover, the 2BGA codes are a generalization of GB codes from cyclic
to more general groups, and share many of the nice properties of the
GB codes.  In particular, we show that 2BGA codes include as a special
case quantum hypergraph-product codes constructed from classical left
(right) group-algebra codes, which guarantees the existence of
finite-rate 2BGA codes with $d=\mathcal{O}(n^{1/2})$.

Although we have not been able to give explicit expressions for the
parameters of 2BGA codes, we constructed a number of equalities and
inequalities relating the parameters to those of other classical and
quantum codes.  From the practical point of view, most important
results are the code equivalence relations in Theorem
\ref{th:permutation-equiv}, and the analysis of block structure of
2BGA codes in Sec.~\ref{sec:block-structure}.

We used these symmetries to enumerate the inequivalent parameters of
binary (designed for qubits) 2BGA codes with stabilizer generator
weights $W\le 8$ for all abelian groups of ranks $\ell\le 50$ and
non-abelian groups of ranks $\ell\le 100$.  Although the sample is too
small to even try to identify the asymptotic form of the distance
scaling, some of the constructed codes have parameters substantially
better than those of GB codes of similar size.  Some of the
constructed codes with row weights $W=8$, distances $d\ge8$ and
dimensions $k\ge 8$ have $kd>n$, a condition difficult to reach for a
short quantum LDPC code.  These codes with larger $k$ have many
redundant minimum-weight stabilizer generators and are expected to
perform well in a fault-tolerant setting as data-syndrome
codes\cite{Fujiwara-2014,Ashikhmin-Lai-Brun-2014,Ashikhmin-Lai-Brun-2016,%
  Zeng-Ashikhmin-Woolls-Pryadko-2019}.

The 2BGA codes based on non-abelian groups have bigger set of possible
parameters than the abelian 2BGA or GB codes; in particular, only the
former codes may have a dimension $k$ given by an odd number, and
there are more short degenerate non-abelian codes with large $k$, see
Table \ref{tab:large-k}, especially for larger $n$.  On the other
hand, some of abelian-group codes found have better parameters than
any of the non-abelian 2BGA codes with the same sizes, and the abelian
codes in Table \ref{tab:large-k} are all obtained from cyclic groups
(and thus are GB codes).  Thus, at least for group sizes studied here,
there is no clear advantage of abelian vs.\ non-abelian 2BGA codes.
Definitely, GB codes and abelian 2BGA codes, because of the simpler
structure, are much more convenient to use.

\begin{acknowledgments} We are grateful to Pavel Panteleev for
  enlightening comments on an early version of this work.  This work
  was supported in part by the APS M. Hildred Blewett Fellowship (HKL)
  and the NSF Division of Physics via the grant 2112848 (LPP).
\end{acknowledgments}

\appendix

\section{Detailed proofs for Sec.~\ref{sec:two-block}}
\label{sec:proofs-III}

\subsection{Proof of Statement \ref{th:rank-defect-zero}}
\begin{proof} It is enough to prove the statement for $E\equiv E_A$;
the statement for $F_A$ is obtained by a similar argument (or by a
transposition).

  First, any idempotent matrix $E=E^2\in M_\ell(F)$ is diagonalizable
by a change of basis, since its minimal polynomial $\mu_E(x)=x(1-x)$
factors into distinct linear terms.  Let $E_A=U^{-1}DU$, $U\in
M_\ell(F)$ an invertible matrix and $D$ a diagonal $(0,1)$ matrix with
a block of $r\equiv \rank E_A=\rank A$ ones in the top left corner.
Use $U$ to transform each square matrix, $A'=UAU^{-1}$, $B'=UBU^{-1}$,
and $E_A'=D$.  Such an invertible transformation preserves both ranks
and commutativity, thus $B'$ must be block diagonal, with two square
blocks of size $r$ and $\ell-r$, and ranks $\rank E_AB$ and
$\rank(I-E_A)B$, respectively.  Similarly, since $DA'=A'$ and their
ranks coincide, the matrix $A'$ has only the first $\rank A$ rows
non-zero and linearly independent.  This gives that the rank of the
product $B'A'$ coincides with that of the first block of $B'$, i.e.,
$\rank BA=\rank E_AB$, giving $\delta_X=0$.
\end{proof}

\subsection{Proof of statement \ref{th:rank-defect-equal}}

\begin{proof} Using the definitions in Eq.~(\ref{eq:sim-transpose}),
write
  $$
  SE_ABS^{-1}= F_A^T B^T= (BF_A)^T.
  $$
  The ranks on the l.h.s.\ and on the r.h.s.\ are $p_\star+\delta_X$ and
$p_\star+\delta_Z$, respectively, which gives $\delta_X=\delta_Z$.
\end{proof}

\subsection{Proof of Statement \ref{th:d-upper-chain}}

\begin{proof} To be specific, we only consider $\mu=L$; the case
$\mu=R$ is similar.  The proof amounts to a demonstration that the set
contributing to the distance (\ref{eq:subsystem-d}) for each
subsequent code is a subset of the previous one.  (a) The additional
row block in matrix $H_X^{(L)}$ (compared to $H_X$) guarantees that
any $Z$-like codeword in $Q_L'$ is also a $Z$-like codeword in
$\lp[a,b]$, but not necessarily the other way.  (b) Any non-trivial
$Z$-codeword $\bs u$ in $Q_L''$ is also a non-trivial codeword ${{\bs
u}\choose \bs 0}$ in $Q_L'$.  (c) Any non-zero codeword $\bs u \in
C_L$ is a $Z$-codeword in $Q_L''$, with an extra row block in $H_L$
fully suppressing the degeneracy.  The condition $E_B\bs u=0$
guarantees that $\bs u\neq0$ cannot be set to zero by adding linear
combinations of the rows of $(H_Z)_L$.
\end{proof}

\subsection{Proof of statement \ref{th:upper-d-block-diagonal}}

\begin{proof} Indeed, since $A$ is block-diagonal with the maximum
block size $m$, we can choose a set of basis vectors
$\mathcal{U}\equiv \{\bs u_1,\bs u_2,\ldots\}$ of the code $C_A^\perp$
so that the support of each vector fits entirely in a single block,
which implies $\wgt(\bs u_j)\le m$.  By the condition, this code
contains a non-zero vector $\bs u$ linearly-independent from the
columns of $B$.  Linear independence can be also written as
$(I-E_B)\bs u\neq0$, see Eq.~(\ref{eq:EA-FA-matrices}).  Since $\bs u$
is a linear combination of the basis vectors in $\mathcal{U}$, at
least one of these satisfies the equation $(I-E_B)\bs u_j\neq0$, which
gives the upper bound in question, $d_Z(A,B^T)\le \wgt \bs u_j\le m$.
\end{proof}

\subsection{Proof of statement \ref{th:d-lower-puncturing}}

The proof is based on the following Lemma (note that the formulation
in the original paper \cite{Zeng-Pryadko-hprod-2020} is missing a
condition; this was corrected in the Erratum).
\begin{lemma}[$Z$-puncturing bound\cite{Zeng-Pryadko-hprod-2020}]
  Consider a stabilizer code $Q = \css(H_X, H_Z)$ with the parameters
  $[[n, k, (d_X, d_Z)]]_q$ and a qudit index set ${\cal V} =
  [n]$. Given a partition into complementary sets
  ${\cal I} \subset {\cal V}$ and
  ${\cal J} ={\cal V} \setminus {\cal I}$, suppose a logical generator
  matrix $L_X$ can be chosen so that none of its $k$ rows is supported
  both in ${\cal I}$ and in ${\cal J}$. Let
  $Q' = \css \biglb((H_X )_{\cal I}, H_Z[{\cal I}]\bigrb)$ and
  $Q'' = \css \biglb((H_X )_{\cal J}, H_Z[{\cal J}]\bigrb)$ be the codes whose $X$
  generator matrices are shortened and $Z$ generator matrices
  punctured to ${\cal I}$ and ${\cal J}$, respectively. Then the $Z$ distances of
  the three codes satisfy the inequality $d_Z\ge \min(d'_Z, d''_Z )$.
\end{lemma}

\begin{proof}[Proof of Statement \ref{th:d-lower-puncturing}] We
  construct the lower bound for $Z$-codewords; CSS symmetry combined
  with the block permutation symmetry gives the other bound.  The
  proof amounts to a demonstration that the condition of the
  $Z$-puncturing bound lemma applies, which relates $d_Z$ to the
  $Z$-distances of the gauge-fixed codes (\ref{eq:Z-punctured}), which
  are known to have the same $Z$-distances as the corresponding
  single-block erasure codes, see Eq.~(\ref{eq:subsystem-d}).  Indeed,
  with $\delta_X=\delta_Z=0$, the total number of independent
  codewords in the original code matches the sum of those for the two
  $Z$-punctured codes.  Thus, we just need to show that any
  non-trivial $Z$-codeword in one of the $Z$-punctured codes can be
  padded with zeros to become a non-trivial codeword of the original
  code, and independent from the codewords coming from the other
  punctured code.  To this end, take $\bs u\in F^\ell$ a non-trivial
  $Z$-codeword from the left $Z$-punctured code,
  \begin{equation} A \bs u =0,\quad \bs u+ B \bs w\neq0\quad
    \forall \bs w\in F^\ell.
    \label{eq:Z-cw}
  \end{equation} Immediately, the pair $\bs c_Z\equiv {\bs u\choose\bs 0}$
  is a $Z$-like codeword in the original two-block code, and, from
  the second part of Eq.~(\ref{eq:Z-cw}), the top component remains
  non-zero when arbitrary linear combinations of rows of $H_Z$ are
  added.  Thus, $\bs c_Z$  is a non-trivial $Z$-codeword, and it is not
  degenerate to any codeword coming 
  from the other punctured code.  This argument is repeated identically
  for the second code in Eq.~(\ref{eq:Z-punctured}), up to an
  interchange of the $A$ and $B$ matrices.
\end{proof}

\subsection{Proof of Eq.~(\ref{eq:dZ-quantum-identity-A})}

\begin{proof} Without limiting generality, take $\mu=L$.  We are going
to show that the set of non-trivial $Z$-codewords of the original
two-block code $Q$ is split without an intersection between those of
the codes $Q_1\equiv \css(H_X^{(L)},H_Z)$ and
$Q_2\equiv\css(H_X,H_Z^{(L)})$.

  Indeed, a non-trivial $Z$-codeword in $Q_1$ or $Q_2$ is also a
non-trivial codeword in $Q$, and a non-trivial codeword in $Q$ is
necessarily a codeword in $Q_2$ (possibly trivial).  Second, the ranks
of the extended matrices $H_X^{(L)}$ and $H_Z^{(L)}$ both equal to
$\ell$; the two codes have dimensions $k_1=k_{\rm S}+\delta_Z$ and
$k_2=k_{\rm S}+\delta_X$, adding up to the dimension
(\ref{eq:two-block-k}) of the code $Q$.

Third, consider a non-trivial codeword
$\bs c_Z\equiv {\bs u\choose \bs v}$ in $Q$.  Suppose it also happens
to be a (necessarily non-trivial) $Z$-codeword in $Q_1$, i.e.,
$(1-E_A) \bs v=0$.  This implies $\bs v=A\bs s$, for some
$\bs s\in F^\ell$, which, in turn, gives $\bs u=B\bs s +(I-F_A)\bs w$,
a trivial codeword in $Q_2$.  On the other hand, $c_Z$ is necessarily
a codeword in $Q_2$, and any linear combination of the columns of
$H_Z^{(L)}$ cannot modify the value of $(1-E_A)\bs v$.  That is, when
this value is non-zero (i.e., $\bs c_Z\not\in Q_{1}$), $\bs c_Z$ is a
non-degenerate codeword in $Q_{2}$.  This completes the dichotomy and
the proof.
\end{proof}

\section{Detailed proofs for Sec.~\ref{sec:construction}}
\label{sec:proofs-IV}

\begin{proof}[Proof of Theorem \ref{th:permutation-equiv}] (i) A group
automorphism is a permutation of group elements preserving the action
of group operation, $\varphi(a)\varphi(b)=\varphi(ab)$, $a,b\in G$.
Further, for any size-$\ell$ permutation matrix $S$ we can write for
$H_X=(A,B)$,
  $$
  S H_X {\bs u\choose \bs v}=(SAS^{-1},SBS^{-1}){S\bs u\choose S\bs
v},
  $$
  and similarly for $H_Z=(B^T,-A^T)$. Using $S^T=S^{-1}$, it is easy
to verify that scalar products and, in particular, the row
orthogonality (\ref{eq:CSS-orthogonality}), are preserved by this
transformation.

  \noindent (ii) This follows from the proof of (i) if we choose the
permutation matrix $S=\LL(\alpha^{-1})\RR(\beta)$ and remember that
$L$ and $R$ matrices commute.

  \noindent (iii) This is proved similarly, by rescaling block
components of ${\bs c}_Z$, $\bs u\to x^{-1}\bs u$, $\bs v\to y^{-1}\bs
v$, and doing a similar weight-preserving transformation of ${\bs
c}_X$.

  \noindent (iv) This also is a consequence of commutativity of left
and right matrices.  Indeed, if we denote $L\equiv \LL(\alpha)$ and
$R\equiv \RR(\beta)$, it is easy to verify that the modified block
matrices in Eq.~(\ref{eq:css-blocks}) are $A'= AL$, $B'= BR$, so that
components of a ${\bs c}_Z$ vector are transformed as $\bs u\to L^T
\bs u$, $\bs v\to R^T\bs v$, and an identical transformation for the
components of a ${\bs c}_X$ vector, $\bs u_X\to L^T \bs u_X$, $\bs
v\to R^T\bs v$.  This obviously preserves the scalar products between
$X$ and $Z$ codewords, and we only need to verify
    \begin{eqnarray*} H_Z' \bs c_X'&=&\left(R^T B^T,-L^T
A^T\right){L^T\bs u_X\choose R^T\bs v_X}\\ &=&R^TL^T\left(B^T,-
A^T\right){\bs u_X\choose \bs v_X}=0.
    \end{eqnarray*}

    \noindent (v) The proof is also similar to (i), except we use the
symmetric permutation matrix $S=P=P^T$ from Eq.~(\ref{eq:L-R-map}),
and, in addition, interchange the blocks, $\bs u_\mu'=P\bs v_\mu$,
$\bs v_\mu'=\pm P\bs u_\mu$, $\mu\in\{X,Y\}$.

    \noindent (vi) After a permutation of the blocks, this is an
immediate consequence of Eq.~(\ref{eq:transposition}).  The second
form follows from (iii) and (v).

\end{proof}

\subsection{Proof of Statement \ref{th:dbl-coset-codes}}

\begin{proof} The equivalence of the two codes is evident from Theorem
\ref{th:permutation-equiv}(ii).  With Eq.~(\ref{eq:cx-orthogonality}),
(\ref{eq:cx-degeneracy}), the corresponding transformation for a pair
of group algebra elements is $[u,v]\to [u,vx^{-1}]$; in particular,
the row $x$ goes to $1$.  Finally, this invertible map sends the
original double coset $G_axG_b$ to $G_a xG_b x^{-1}=G_a 1
G_{xbx^{-1}}$.
\end{proof}

\subsection{Proof of Statement \ref{th:triple-product}}

\begin{proof} The subgroup $N$ being normal both in $G_a$ and $G_b$
guarantees that we can decompose $G_a=H_a \rtimes N$ and $G_b =
N\ltimes H_b$ as semidirect products\cite{Bechtell-book-1971}, where
$H_a= G_a/N$ and $H_b= N\backslash G_b$ are sets of cosets.  The
semidirect product, e.g., with the normal group on the right, is
defined as a group with elements from the set of all pairs $H_a
\rtimes N =\{(h,\gamma)|h\in H_a, \gamma\in N\}$ and the group
product
  $$
  (h_1,\gamma_1)\cdot (h_1,\gamma_2) \equiv (h_1h_2,
h_2^{-1}\gamma_1h_2\,\gamma_2).
  $$
  Similarly, any element of the double coset $G_a 1 G_b$ can be
written as a triplet, $(\alpha,\gamma,\beta)$, with $\alpha\in G_a$,
$\gamma\in N$, and $\beta\in G_b$, with all triplets in the form
$(\alpha x,x^{-1}\gamma y^{-1},y\beta)$, $x,y\in N$ united into
product-preserving equivalence classes.  The group product is defined
as
  $$
  (\alpha_1,\gamma_1,\beta_1)\cdot (\alpha_2,\gamma_2,\beta_2)
=(\alpha_1\alpha_2,\alpha_2^{-1}\gamma_1\alpha_2\,
\beta_1\gamma_2\beta_1^{-1},\beta_1\beta_2),
  $$
  where the elements from $H_a$ and $H_b$ are forced to commute, and
the abelian property of $N$ is used to ensure the consistency of the
definition.  It is easy to verify the group axioms: thus defined
product is associative, the identity element is the equivalence class
of $(1,1,1)$, and the inverse of $(\alpha,\gamma,\beta)$ is
$(\alpha^{-1}, \alpha \beta^{-1}\gamma^{-1}\beta
\alpha^{-1},\beta^{-1})$, which is both a left and a right inverse.
Finally, this map also gives a natural map for the group algebra
elements $a$ and $b$, with the multiplication by an element of $G_a$
from the left or an element of $G_b$ from the right giving the
expected results.
\end{proof}

\subsection{Proof of statement \ref{th:abelian}}

\begin{proof}[Proof of statement \ref{th:abelian}] As discussed in the
previous section, the square matrices of 2BGA codes under
consideration here have the form of Kronecker products, $A=A_1\otimes
I_{m_b}$, $B=I_{m_a}\times B_1$, where $m_a$ and $m_b$ are indices of
the support groups in $G$, with $\ell=m_am_bc$ and $c\equiv |N|$.
Further, $A_1$ and $B_1$ have square blocks of size $c$ which can be
readily seen to have the form of group algebra matrices
$\LL_{N}(x_{ij})=\RR_{N}(x_{ij})$ and $x_{ij}\in F[N]$.  That is, the
original 2BGA code is an abelian LP code, which can also be seen as an
HP code over the ring $R=F[N]$. As explained in the Appendix of
Ref.~\onlinecite{Panteleev-Kalachev-2020}, any such code can be
decomposed further as a direct sum of HP codes over cyclic rings,
$F[C_{\ell_i}]$, $1\le i\le s$, where the groups $C_{\ell_i}$ are
those in the decomposition of finite abelian group $N=C_{\ell_1}\times
C_{\ell_2}\times \ldots \times C_{\ell_s}$ into a direct product of
cyclic groups.  Each of these HP codes can be also seen as an
$F$-linear quasicyclic LP code, and as a special case of a
hyperbicycle code\cite{Kovalev-Pryadko-Hyperbicycle-2013}.
Importantly, at each step, we can reconstruct matrices over the
original field $F$ as block diagonal matrices, with the blocks given
by the corresponding matrices of the quasicyclic codes in the
decomposition; these transformations preserve the total matrix rank
and the dimension $k$ of the original code.

The final step of the proof is to use Smith normal form (SNF)
decomposition over the polynomial ring $F[x]$ to show that a
quasicyclic LP code can be further decomposed as a direct sum of GB
codes, see Lemma \ref{th:QC-HP} below.  Since matrix ranks are
additive, and rank defects vanish for GB codes, this proves
$\delta_X=\delta_Z=0$ for 2BGA codes under consideration, and also for
all quasiabelian LP codes.
\end{proof}

\begin{lemma}[Decomposition of a quasicyclic LP code]
  \label{th:QC-HP} Given $R\equiv F[x]/(x^\ell-1)$, a ring of modular
  polynomials isomorphic to circulant matrices over the finite field
  $F$, let $A$ and $B$ be arbitrary matrices over $R$.  The code
  $\lp[A,B]$ is isomorphic a direct sum of GB codes.
\end{lemma}
As a reminder, the quasicyclic LP
code\cite{Panteleev-Kalachev-2019,Panteleev-Kalachev-2020} is
constructed by replacing matrix elements of the associated HP code
over $R$, polynomials $h(x)\in R$, with the cyclic permutation
matrices $h(P)$, see Eq.~(\ref{eq:permutation-matrix}).  The
transformation in the proof below is done with the help of SNF
decomposition.  In general, this requires a non-trivial basis change.
That is, only the dimensions of the corresponding spaces are
preserved, but not the code distances.

\begin{proof}[Proof of Lemma \ref{th:QC-HP}]
  Denote the dimensions of matrices $A$ and $B$, respectively, as
  $r_A\times n_A$ and $r_B\times n_B$.  Then the associated HP code
  over $R$ has CSS generators
\begin{equation}
  \label{eq:qc-HP-matrices}
  H_X=\left(A\otimes I_B,I_A\otimes B\right),
  \quad H_Z^T={I_A'\otimes B\choose
    -A\otimes I_B'},
\end{equation}
where $I_A$, $I_B$, $I_A'$, and $I_B'$, respectively, are identity
matrices of dimensions $r_A$, $r_B$, $n_A$, and $n_B$.
The matrix elements of the original matrices $A$ and $B$ are
polynomials from $R$.  Considering the polynomials as elements of
$F[x]$, a principal ideal domain, SNF of such matrices can be readily
constructed using elementary row and column transformations, e.g.,
$A=U_A D_A V_A$, where $U_A$ and $V_A$ are square matrices of size
$r_A$ and $n_A$ with unit determinants, respectively, and
$D_A=\diag\biglb(a_1(x), a_2(x), \ldots \bigrb)$, where $a_j(x)$ along
the diagonal are \emph{SNF invariants}, with each subsequent
polynomial divided by the previous one, $a_j(x)\mid a_{j+1}(x)$,
$0<j<\min(r_a,n_a)$.  The SNF over $F[x]/(x^\ell-1)$ is obtained by
taking these matrices modulo $x^\ell-1$ element-wise.  This preserves
the unit determinants of the matrices $U_A$ and $V_A$, i.e., these
matrices remain invertible.  Denoting the SNF invariants of the matrix
$B$ as $b_j(x)$, $0<j\le\min(r_B,n_B)$, it is easy to see that the
invertible matrices can be factored out in the CSS generator matrices
(\ref{eq:qc-HP-matrices}), e.g., {\small
  \begin{eqnarray*} \lefteqn{ (U_AD_AV_A\otimes I_B, I_A\otimes U_BD_BV_B) }
    & &  \\
    & &\hskip-1em =(U_A\otimes U_B)\,(D_A\otimes I_B, I_A\otimes D_B) \left(
        \begin{array}[c]{cc}
          \!V_A\otimes U_B^{-1}\!\!& \\
                                   &\!\! U_A^{-1}\otimes V_B\!
        \end{array} \right),
  \end{eqnarray*}}%
which preserves the orthogonality between the rows of the transformed
matrices $H_X'$ and $H_Z'$.  Evidently, the transformed matrices are
constructed similarly to Eq.~(\ref{eq:qc-HP-matrices}),
but from the diagonal matrices $D_A$ and $D_B$, so that each row contains
just two polynomials, e.g., the row $i+(j-1)r_A$ of
$H_X$ has $a_i(x)$ in the first block and $b_j(x)$ in the second.
This gives a block-diagonal form of the
original HP code, with individual blocks forming the codes
$\gb[a_i(x),b_j(x)]$ constructed from all pairwise combinations of SNF
invariants of the original matrices $A$ and $B$.  In particular, this
gives the dimension of the original quasicyclic LP code as
\begin{equation}
  k=2\sum_{i=1}^{\min(r_A,n_A)}\sum_{j=1}^{\min(r_B,n_B)}\gcd(a_i(x),b_j(x),x^\ell-1).
  \label{eq:gb-decomposition}
\end{equation}
Unlike the corresponding expressions in Appendix B of
Ref.~\onlinecite{Panteleev-Kalachev-2020}, this formula does not
require that the group algebra $F[C_\ell]$ be semisimple.
\end{proof}

\subsection{Proof of Statement \ref{th:semisimple}}

\begin{proof} The result follows from Statement
  \ref{th:rank-defect-zero}.  Indeed, semisimple ideals $aR$ and $Ra$
  are summands in $R$, and can be generated by idempotents $e_a$ and
  $f_a$, respectively, such that $a=e_a a=af_a$. The conditions of
  Statement \ref{th:rank-defect-zero} are satisfied by taking
  $E_A=\LL(e_a)$ and $F_A=\LL(f_a)$ which necessarily commute with
  $B=\RR(b)$.
\end{proof}

\subsection{Proof of Statement \ref{th:lower-d-central-intersection}}
\label{proof:lower-d-central-intersection}

The proof is based on the following
\begin{lemma}[Trivial quasi-abelian LP codes]
  \label{th:qc-lp-triv}
  Given an abelian group $N$ and a finite field $F$, consider the
  abelian group algebra $R\equiv F[N]$.  Let $A$ and $B$ be matrices
  with elements in $R$ such that the classical codes $C_A^\perp$ and
  $C_B^\perp$ both have zero dimensions.  Then the quasiabelian code
  $\lp[A,B]$ is trivial, $\dim \lp[A,B]=0$.
\end{lemma}
\begin{proof}
  This result is proved similarly to Statement \ref{th:abelian}.
  Namely, we start from the decomposition of the LP code as a direct
  sum of quasicyclic LP codes (see Appendix B in
  Ref.~\onlinecite{Panteleev-Kalachev-2020}), combined with the
  decomposition in Lemma \ref{th:QC-HP}.  With both $C_A^\perp$ and
  $C_B^\perp$ trivial, the SNF invariants of both matrices must all be
  unit.  As a result, every GB code in the decomposition of Lemma
  \ref{th:QC-HP} is trivial, which gives the result immediately.

  Alternatively, we can say that all-unit SNF invariants of $A$ and
  $B$ imply that these matrices have well defined ranks over $R$, and
  follow the conventional rank-based derivation for the dimension of
  HP code\cite{Tillich-Zemor-2009}.  It gives zero when both matrices
  have full row ranks.
\end{proof}

\begin{proof}[Proof of Statement \ref{th:lower-d-central-intersection}]
  The proof goes along the lines of that for the lower bound on the
  distance of conventional HP codes\cite{Tillich-Zemor-2009}; it is
  based on Lemma \ref{th:qc-lp-triv}.  In this proof, matrices over
  $R$ are labeled by capital letters in the usual math italic font,
  while the corresponding matrices over $F$ are labeled in bold
  italic, e.g, $A$ with matrix elements $a_{ij}\in R$ and $\bs A$
  formed by blocks $\LL_N(a_{ij})=\RR_N(a_{ij})$.  We also denote
  $\ell_a\equiv [G_a:N]$ and $\ell_b\equiv [G_b:N]$ the indices of the
  intersection group $N$ in the two support subgroups, so that
  $\ell=c\ell_a\ell_b$.

  The 2BGA code $\lp[a,b]$ is a two-block code constructed from
  matrices $\bs A=\bs A_1\otimes \bs I_{\ell_b}$ and
  $\bs B=\bs  I_{\ell_a}\otimes \bs B_1$, where $\bs A_1$ and $\bs B_1$,
  respectively, are block matrices equivalent to
  $A_1\in M_{\ell_a}[R]$ and $B_1\in M_{\ell_b}[R]$.  Thus, the
  original code is equivalent to the $R$-linear code $\hp[A_1,B_1]$,
  and it is this equivalence that is used to construct the lower
  distance bound.

  Given a vector $\bs e\in F^{2\ell}$ orthogonal to the rows of the
  generator matrix $\bs H_X$ of the 2BGA code $\lp[a,b]$ equivalent to
  $\hp[A_1,B_1]$, we construct sets ${\cal I}_A\subset [\ell_a]$ and
  ${\cal I}_B\subset [\ell_b]$ indexing only the columns of the
  matrices $A_1$ and $B_1$ incident on non-zero elements of $\bs e$ in
  the product $\bs H_X \,\bs e=0$, the sets
  ${\cal I}_A'$ and ${\cal I}_B'$ labeling all columns in the
  corresponding blocks of $\bs A_1$ and $\bs B_1$, and the set
  ${\cal I}={\cal I}_A'\times [\ell_b] \,{\textstyle\bigsqcup}\,
  [\ell_a]\times {\cal I}_B'$ labeling all such columns in $\bs H_X$,
  a disjoint union of the corresponding sets in the left and in the
  right blocks.  Each element of $R$ corresponds to a block of size
  $c\equiv |N|$ in matrices $\bs A_1$ and $\bs B_1$, thus
  $$ |{\cal I}_\mu'|=c\,|{\cal I}_\mu|\le c \wgt (\bs e),\quad \mu\in \{A,B\}.$$
  Denote $\bs H_X'$, $\bs H_Z'$ the CSS generator matrices of the LP
  code constructed from punctured matrices $A_1[{\cal I}_A]$ and
  $B_1[{\cal I}_B]$.  By construction, the shortened vector
  $\bs e[{\cal I}]$ is a $Z$-like codeword in the modified LP code,
  $\bs H_X' \,\bs e[{\cal I}]=0$.  On the other hand, if
  $|{\cal I}_A'|<d_A^\perp$ and $|{\cal I}_B'|< d_B^\perp$, the
  modified LP code must be trivial by Lemma \ref{th:qc-lp-triv}, i.e.,
  $\bs e[{\cal I}]$ can only be a trivial codeword, and thus a linear
  combination of the rows of $\bs H_Z'$.  These latter rows can be
  constructed by shortening a subset of the rows of the original
  matrix $\bs H_X$ (where we drop only positions equal to zero in each
  row of the subset), thus the full vector $\bs e$ is a linear
  combination of the rows of the original matrix $\bs H_Z$.  The
  inequality on the subset sizes is satisfied whenever
  $c\wgt(\bs e)<\min(d_A^\perp,d_B^\perp)$, which proves that the
  distance of the original LP code satisfies
  $d_Z\ge \min(d_A^\perp,d_B^\perp)/c$, and, since $d_Z$ is an
  integer, $d_Z\ge d_0$ as stated.
\end{proof}

\subsection{Proof of Statement \ref{th:upper-d-central-intersection}}

\begin{proof}
  If the code $C_A^\perp\cap C_{J}$ is trivial, its distance is
  infinite, and the upper bound in question is definitely satisfied.
  Assuming otherwise, take any non-zero vector
  $\bs u\in C_A^\perp\cap C_{J}$; the corresponding pair
  ${\bs u\choose \bs 0}$ is clearly a $Z$-codeword in the 2BGA code,
  and we just need to verify that it is not degenerate to a zero
  vector.

  The minimum-weight elements in $C_A^\perp $ are associated with a
  single block of $A\equiv \LL(a)$, e.g., the terms in
  Eq.~(\ref{eq:dbl-coset-decomposition}) with $\beta=1$, the element
  of $F[G]$ corresponding to $\bs u$ has the form
  $u=\sum_{\alpha\in \cal A}\alpha \,u_\alpha$; intersection with
  $C_{J}$ ensures that all $u_\alpha\in \cal I$.  The code
  $C_{B^T}^\perp \cap \widehat{C}_{J}$ contains vectors
  $\bs z\equiv P\bs y$, such that $\bs y\in C_{J}$ and
  $B^T \bs z=0$, where $B\equiv \RR(b)$ and $P$ is the permutation
  matrix in Eq.~(\ref{eq:L-R-map}).  Rewrite the condition
  $B^T\bs z=0$ in terms of the group algebra element associated with
  $\bs y$; with the help of Eq.~(\ref{eq:L-R-map}) it reads
  $$ 0=\RR(b)^T P \LL(y)=P\LL(b)\LL(y)=P \LL(by),$$
  or simply $by=0$.  Again, block structure of $B$ and group symmetry
  guarantees that we can choose $y$ in the form
  $y=\sum_{\beta\in\mathcal{B}} y_\beta \beta$.  By construction,
  $y_\alpha\in {J}$, a maximal ideal, which ensures that
  $uy=\sum_{\alpha}\sum_\beta \alpha u_\alpha y_\beta \beta\neq0$, and
  thus for any $w\in F[G]$, $(u-w b)y=uy\neq0$, which guarantees that
  the pair ${\bs u\choose 0}$ be a non-degenerate $Z$-codeword in
  $\lp[a,b]$.
\end{proof}

\section{Additional examples}
\label{sec:examples}

Table \ref{tab:Cmh-2+6} gives explicitly the group algebra elements
for constructing abelian 2BGA codes from the sequences $kd=n$.
Namely, for a given group
$C_{mh}=C_m\times C_2=\langle x,s| x^m=s^2=xsx^{-1}s^{-1}=1\rangle$,
$m\ge1$, only the maximum-distance codes with $k/2$ a factor of $n$
are shown.  With polynomial decomposition $a=a_0(x)+sa_1(x)$ and
$b=b_0(x)=sb_1(x)$, these codes can be also seen as index-4 qQC
two-block codes constructed from the circulant matrices
$$
A=\left(
  \begin{array}[c]{cc}
    a_0(x)&a_1(x)\\a_1(x)&a_0(x)
  \end{array}\right), \quad
B=\left(
  \begin{array}[c]{cc}
    b_0(x)&b_1(x)\\b_1(x)&b_0(x)
  \end{array}\right).
$$
\begin{table*}[htbp]
  \centering
  \begin{tabular}[c]{c|c||c|c|c||c|c}
    $m$&$\ell$& $n$&$k$&$d$&$a$&$b$ \\ \hline
4 & 8&16& 2 & \bf 4 &$1+x $   &$1+ x+ s+ x^2+ sx+ sx^3 $ \\
  &  &  & 4 & 4 &$1+x $   &$1+ x+ s+ x^2+ sx+ x^3 $ \\
  &  &  & 8 & 2 &$1+s $   &$1+ x+ s+ x^2+ sx+ sx^2 $\\ \hline
6 &12 &24 & 4 & \textbf{5}  &$1+ x $   &$1+ x^3+ s+ x^4+ x^2+ sx $\\
  &  &   & 12& 2  &$1+ x^3 $   &$1+ x^3+ s+ x^4+ sx^3+ x $\\ \hline
8& 16& 32 & 8 & 4   &$1+ x^6 $   &$1+ sx^7+ sx^4+ x^6+ sx^5+ sx^2 $\\
 &   &  &16 & 2   &$1+ sx^4 $   &$1+ sx^7+ sx^4+ x^6+ x^3+ sx^2 $\\    \hline
10 & 20 &40 & 4 & \textbf{8}     &$1+ x $   &$1+ x^5+ x^6+ sx^6+ x^7+ sx^3 $\\
   &    &   & 8 & 5    &$1+ x^6 $   &$1+ x^5+ s+ x^6+ x+ sx^2 $\\
   &    &   & 20& 2   &$1+ x^5 $   &$1+ x^5+ s+ x^6+ sx^5+ x $\\    \hline
12 & 24 & 48 & 8  &6    &$1+ sx^{10} $   &$1+ x^3+ sx^6+ x^4+ x^7+ x^8 $\\
   &    &   & 12 & 4   &$1+ x^3 $   &$1+ x^3+ sx^6+ x^4+ sx^9+ x^7 $\\
   &    &   & 16 & 3   &$1+ x^4 $   &$1+ x^3+ sx^6+ x^4+ x^7+ sx^{10} $\\
   &    &   & 24 & 2   &$1+ sx^6 $   &$1+ x^3+ sx^6+ x^4+ sx^9+ sx^{10}
                                       $\\ \hline
14 &28&56  & 4 & {\bf10}    &$1+ x $   &$1+ x^7+ sx^8+ x^2+ x^3+ sx^{11} $\\
   & &  & 8 & 7    &$1+ x^8 $   &$1+ x^7+ s+ x^8+ x^9+ sx^4 $\\
   & &  & 28&  2  &$1+ x^7 $   &$1+ x^7+ s+ x^8+ sx^7+ x $
  \end{tabular}
  \caption{Largest-distance 2BGA codes from abelian groups
    $C_{mh}=C_m\times C_2$, with $k$ a factor of $n=4m$ and
    $W_a=2$, $W_b=6$. The group generators are $x$ and $s$, with
    $x^m=1$, $s^2=1$, and $xs=sx$.  The distances which \emph{fail} to
    satisfy the condition $kd=n$ are given in bold.}
  \label{tab:Cmh-2+6}
\end{table*}

Table \ref{tab:D2m-2+6} gives explicitly the group algebra elements
for constructing non-abelian 2BGA codes with $kd=n$.  All codes are
constructed from the groups
$D_{m}=C_m\ltimes C_2=\langle r,s| r^m=s^2=(rs)^2=1\rangle$, $m\ge1$.
With polynomial decomposition $a=a_0(r)+sa_1(r)$ and $b=b_0(r)+sb_1(r)$, these
codes can be also seen as index-4 qQC two-block codes constructed from
circulant matrices, 
$$
A=\left(
  \begin{array}[c]{cc}
    a_0(x)&\overline{a_1(x)}\\ {a_1(x)}&\overline{a_0(x)}
  \end{array}\right), \quad
B=\left(
  \begin{array}[c]{cc}
    b_0(x)&\overline{b_1(x)}\\b_1(x)&\overline{b_0(x)}
  \end{array}\right),
$$
where
$\overline{a_0(r)}=a_0(r^{-1})\equiv s a_0(r)s$ is the reverse of the
polynomial $a_0(r)$.

\begin{table*}[htbp]
  \centering
  \begin{tabular}[c]{c|c||c|c|c||c|c}
    $m$&$\ell$& $n$&$k$&$d$&$a$&$b$ \\ \hline
6 &12 &24 & 8  &3     &$1+ r^4 $   &$1+ s r^4+ r^3+ r^4+ s r^2+ r $\\
  &   &   & 12 & 2    &$1+ r^3 $   &$1+ s r+ r^3+ r^4+ s r^4+ r $  \\     \hline 
8 &16 &32 & 8  &4     &$1+ r^2 $   &$1+ s r^5+ s r^4+ r^2+ s r^7+ s r^6 $\\
  &   &   & 16 & 2    &$1+ r^4 $   &$1+ s r^3+ s r^6+ r^4+ s r^7+ s r^2 $\\    \hline 
9 &18 &36 & 12 & 3    &$1+ r^3 $   &$1+ s+ r+ r^3+ s r^3+ r^4 $    \\    \hline 
10&20 & 40& 8  &5     &$1+ r^2 $   &$1+ s r^4+ r^5+ r^2+ s r^6+ r $ \\
  &   &   & 20 & 2    &$1+ r^5 $   &$1+ s r^2+ r^5+ r^6+ s r^7+ r $ \\    \hline 
12&24 & 48& 8  &6     &$1+ r^{10} $   &$1+ s r^8+ r^9+ r^4+ s r^2+ r^5 $\\
  &   &   & 12 & 4    &$1+ r^3 $   &$1+ s r^7+ r^3+ r^4+ s r^{10}+ r^7 $\\
  &   &   & 16 & 3    &$1+ r^8 $   &$1+ s r^8+ r^9+ r^8+ s r^4+ r^5 $\\
  &   &   & 24 & 2    &$1+ r^6 $   &$1+ s r^{11}+ r^6+ s r^5+ r+ r^7 $\\    \hline 
14&28 &56 & 8  &7     &$1+ r^4 $   &$1+ s r^{11}+ r^7+ s r^5+ r^{12}+ r^9 $\\
  &   &   & 28 & 2    &$1+ r^7 $   &$1+ s r^2+ r^7+ r^8+ s r^9+ r $\\    \hline 
15&30 & 60& 12 & 5    &$1+ r^{12} $   &$1+ s r^{14}+ r^5+ r^{12}+ s r^{11}+ r^{14} $\\
  &   &   & 20 & 3    &$1+ r^5 $   &$1+ s r^{13}+ r^5+ r^{12}+ s r^3+ r^2 $\\    \hline 
16&32 & 64&  8 & 8    &$1+ r^6 $   &$1+ s r^{12}+ s r^9+ r^6+ s+ s r $\\
  &   &   &  16&  4   &$1+ r^4 $   &$1+ s r^{10}+ s r^3+ r^4+ s r^{14}+ s r^7 $\\
  &   &   &  32&  2   &$1+ r^8 $   &$1+ s r^{11}+ s r^{12}+ r^8+ s r^3+ s r^4 $\\    \hline 
  \end{tabular}
  \caption{As in Table \ref{tab:Cmh-2+6} but for non-abelian dihedral groups
    $D_{m}=\langle r,s|r^m=s^2=(rs)^2=1\rangle$, with $k$ a factor
    of $n=4m$.  Parameters of all codes listed satisfy the condition
    $kd=n$.}
  \label{tab:D2m-2+6}
\end{table*}

\bibliography{lpp,qc_all,more_qc,ldpc,linalg,teach}

\end{document}